\documentclass{article}
\usepackage{latexsym}
\newtheorem{theorem}{Theorem}
\newtheorem{definition}{Definition}
\newtheorem{example}{Example}
\newtheorem{corollary}{Corollary}

\newtheorem{proposition}{Proposition}
\newtheorem{lemma}{Lemma}

\newcommand{\blackslug}{\mbox{\hskip 1pt \vrule width 4pt height 8pt 
depth 1.5pt \hskip 1pt}}
\newcommand{\qed}{\quad\blackslug\lower 8.5pt\null\par\noindent}
\newenvironment{proof}{\par\noindent{\bf Proof:}}{\qed \par}

\newcommand{\cR}{\mbox{${\cal R}$}}

\title{Revealed Preferences for Matching with Contracts}
\author{Daniel Lehmann\\
The Rachel and Selim Benin School \\of Computer Science and Engineering, \\Hebrew University, 
\\Jerusalem 91904, Israel
\\lehmann@cs.huji.ac.il
}
\date{March 2020} 
\begin{document}
\maketitle
\begin{abstract}
Many-to-many matching with contracts is studied in the framework of revealed preferences.
All preferences are described by choice functions that satisfy natural conditions.
Under a no-externality assumption individual preferences can be aggregated 
into a single choice function expressing a collective preference.
In this framework, a two-sided matching problem may be described 
as an agreement problem between two parties, 
each of which is characterized by a choice function over the set of contracts: 
the two parties must find a {\em stable agreement}, i.e., 
a set of contracts from which no party will want to take away any contract 
and to which the two parties cannot agree to add any contract.
On such stable agreements a party's preference relation is a partial order 
and the two parties have inverse preferences.
An algorithm is presented that generalizes algorithms previously proposed
by Gale-Shapley, Roth, Kelso-Crawford, Alkan-Gale, Hatfield-Milgrom and Hatfield-Kominers
in less general situations. 
To any agreement problem, this algorithm provides a stable agreement that is preferred
to all stable agreements by one of the parties and therefore less preferred than all stable agreements by 
the other party.
In this algorithm the agents' preferences are expressed only by their choice functions and therefore no notion 
of utility is assumed.
The number of steps of the algorithm is linear in the size of the set of contracts, i.e., 
polynomial in the size of the problem.
The algorithm provides a proof that stable agreements form a lattice under the 
two inverse preference relations.
Under additional assumptions on the role of money in preferences, 
agreement problems can describe general two-sided markets in which goods are exchanged for money.
Stable agreements provide a solution concept, including prices, that is more general than competitive equilibria.
Stable agreements exist in many situations in which there is no competitive equilibrium.
They satisfy an almost one price law for identical items.
Only part of those results have been presented, with a different proof, in~\cite{Choice_and_Matching}.
\end{abstract}

\section{Introduction} \label{sec:intro}
Matching is a topic of intense research activity and has provided many deep and beautiful
theoretical results, solutions to practical problems of great importance and original 
conceptual insights into the description of preferences, equilibrium theory, social choice theory
and the conception of auction mechanisms. 
It is not possible to survey here even a small part of this activity.
It is only those results directly related to the present work that will be mentioned below. 

Gale and Shapley~\cite{GaleShapley:marriage} described one-to-one and a restricted form of
one-to-many matching problems. 
They showed that the existence of a stable solution was guaranteed for all possible preferences
of the agents. 
They also showed that, in the one-to-one situation there is, for each of the two sides, 
a stable solution that is optimal for it and that, in the one-to-many situation there is such a stable
solution that is optimal for the side interested in a single agent (the doctors' side).
In~\cite{Knuth:mariages} Knuth attributes to J.H. Conway the remark that the stable solutions
to the one-to-many matching problems above form a lattice under a suitable partial ordering.
Kelso and Crawford~\cite{KelsoCraw:82} considered a more general situation:
the job market, matching firms with teams of workers that receive a salary, 
whereas a worker can work for only one firm.
A discussion of a possible lattice structure for stable solutions to the Kelso-Crawford problem 
can be found in~\cite{Roth:84, Roth_MOR:85}.
Blair~\cite{Blair:88} considered a generalized job market in which workers can contract with
many firms and showed that, even in such a many-to-many matching situation 
the stable solutions present a lattice structure.
The study of many-to-many matchings has been actively pursued since, see, for example,
\cite{Echenique_Oviedo:2006}.

All the studies mentioned above assume a full description of the preferences of the agents,
by a total ordering on the set of all subsets of alternatives.
The stream of research in economic theory, 
initiated by Samuelson~\cite{Samuelson:38, Samuelson:48}, developed by~\cite{Sen:revealed}
and brought to matching theory in~\cite{Roth:84, Roth_MOR:85, Alkan_Gale:2003} 
proposed to assume only preferences revealed by the agents' actions.
The agents' preferences are represented by choice functions that pick, 
out of a set of possible choices, the one {\em preferred}.

In their seminal article~\cite{Hatfield_Milgrom:2005}, Hatfield and Milgrom 
proposed a fundamental shift of focus.
Instead of focusing on the preferences that agents on one side have for agents on the other
side (or for groups of agents on the other side) they suggested that one should focus on the
preferences of the agents on subsets of a set of contracts, where each contract links two
agents on opposite sides. They achieved two goals: 
they broke out of the limited framework of the job market in which only a single salary
parameter was considered, and they opened the way for considering preferences 
that were not preferences on the agents of the other side. 
For them, choice functions that satisfy a {\em substitutes} property describe 
the preferences of hospitals over sets of doctors, not utility functions.
Many researchers pursued the paths opened by~\cite{Hatfield_Milgrom:2005} 
and many of them will be mentioned below.

\section{About this paper} \label{sec:this}
\subsection{Goals and summary} \label{sec:goals}
This paper's goal is to generalize the one-to-many matching with contracts situation
of~\cite{Hatfield_Milgrom:2005} to a many-to-many situation and to study the structure
of the stable solutions in such a situation.
It can be compared to \cite{Echenique_Oviedo:2006}, which does not consider contracts.

In this paper, preferences of the agents over sets of contracts are specified by choice functions
without any assumption on the structure of the set of contracts.
A set of three conditions on choice functions guarantees 
\begin{enumerate}
\item that the aggregation of the choice functions of all agents on the same side satisfies 
the same conditions and represents the collective preferences of the side in question, 
thus reducing the many-to-many situation to a one-to-one situation, 
\item that a natural partial order can be defined on stable solutions, 
\item that the set of stable solutions is always a non-empty lattice under this partial order, and
\item that a simple, elegant algorithm, expressed in terms of two choice functions,
that generalizes all previously described matching algorithms, 
always finds an extremal stable solution.
\end{enumerate}
General two-sided markets can be described in this framework by postulating that consumers prefer low prices
and producers high prices, and that none of them cares with whom he or she trades.
A {\em Law of two prices} is proved: identical goods are traded at an essentially uniform price.

\subsection{Plan of this paper} \label{sec:plan}
This paper develops a significant mathematical apparatus, but its purpose is mainly 
conceptual. Therefore it stresses the change of perspective it proposes on matching and
delays the mathematical results as much as possible.
Section~\ref{sec:preferences} discusses the representation of preferences over subsets
of a given set.
Section~\ref{sec:choice_functions} defines the set of coherent choice functions that
will be used to represent preferences.
It explains the three properties appearing in the definition, surveys the literature
and presents examples of coherent choice functions.
Section~\ref{sec:bilateral}, describes many-to-many matching when the agents' preferences 
are specified by coherent choice functions. It defines the collective preferences of each side
and shows that those are also coherent.
Section~\ref{sec:agreement} redefines many-to-many matching with coherent choice functions
as an agreement problem between two parties each of them equipped 
with a coherent choice function.
It proposes a solution concept: stable agreements, that generalizes the stable solutions
of the one-to-many matching problems.
Section~\ref{sec:algorithm} presents an algorithm, in terms of choice functions, that 
finds a stable agreement.
It shows that it reduces to the Gale-Shapley differed acceptance algorithm in the marriage problem
and that it is a polynomial-time algorithm.
Section~\ref{sec:properties_algo} is the technical part of this paper and 
is devoted to the proof of the properties of the algorithm defined 
in Section~\ref{sec:algorithm}. 
It includes a proof of the lattice structure of the set of stable agreements.
Section~\ref{sec:money} generalizes the framework of many-to-many matching with contracts
to economies of producers and consumers in which money is present. 
It shows that stable agreements define prices for items.
Section~\ref{sec:open} describes open problems and future work.
In particular it suggests that many-to-many matching may be the right framework to 
understand imperfect finite bilateral markets.
Section~\ref{sec:conclusion} is a conclusion.

\section{The representation of preferences} \label{sec:preferences}
Preferences, or incentives, are at the center of economic theory and a lot of effort
has been devoted to understanding and formalizing them.
In situations where the preferences are on the elements of an unstructured set, preferences 
are described in one of two ways:
\begin{enumerate}
\item a quantitative description that associates a number, its value or its utility, 
to each element, or
\item a qualitative description by a binary relation between elements that describes 
which elements are preferred to which.
\end{enumerate}
Any quantitative description induces a natural total and transitive relation, i.e., 
a total pre-order, on the elements.
Qualitative descriptions are almost always taken to be total pre-orders.
Any total pre-order can be defined by a utility function, and therefore one can relatively 
easily translate from one description to the other.

The picture is quite different when preferences are on a structured set. 
In this work, following Hatfield and Milgrom's~\cite{Hatfield_Milgrom:2005}, 
preferences are on the subsets of a base set,
and preferences are assumed to satisfy certain conditions in respect to the subset structure.
In~\cite{GaleShapley:marriage} the authors present the agents' preferences in the marriage 
problem as preferences on the items of an unstructured set: 
boys have preferences over the set of girls and reciprocally, but this is, in fact, not the most
natural description. 
To express the fact that a solution is a complete matching, i.e., that nobody stays unmatched, 
the authors have to impose this restriction as a hard constraint: partial matchings
will not be considered, whereas, had they chosen to describe preferences of, say boys, 
as a preference over sets of girls, they could have formulated their restriction as 
a soft constraint by assuming that boys prefer singletons over all other subsets, and could have 
shown that their solution indeed is a complete matching.
As soon as one considers one-to-many matchings one cannot avoid considering preferences
over subsets, e.g., firms have preferences over sets of workers.

In this work we shall assume that preferences are over the collection of all subsets
of a set $X$. The elements \mbox{$x \in X$} will be called {\em contracts}.
The set $X$ is the set of all possible contracts.
One may describe such preferences in a quantitative way by a valuation
\mbox{$v : 2^{X} \rightarrow \cR$} that associates a real number to any subset of $X$.
This is, for example, the way Kelso and Crawford proceed in~\cite{KelsoCraw:82}.
They defined a family of such valuations, the gross-substitutes valuations, 
and proved their results under the assumption that preferences are gross-substitutes.
Since there is, in their framework, no place for a gross-net distinction, 
this family of valuations will be referred to as {\em v-substitutes} in this paper.
There is a vast literature on this class of functions and the interested reader will find recent
surveys and references in~\cite{RPL:GEB, Lehmann:ultra}.
A Walrasian equilibrium is guaranteed if all agents in a market have preferences described by
v-substitutes valuations.

To describe preferences for matching in a qualitative way, one could look for a suitable family
of binary relations on $2^{X}$.
But the revealed preferences method suggests another formalization: describe an agent's 
preferences by a function \mbox{$f : 2^{X} \rightarrow 2^{X}$} that associates any
subset $A$ of $X$ to the subset of $A$ that the agent prefers to all subsets of $A$.
In Section~\ref{sec:pre-order} we shall define a transitive binary relation associated
with such a function.
Hatfield and Milgrom~\cite{Hatfield_Milgrom:2005} use such choice functions satisfying
the Substitutes property of Definition~\ref{def:coherent} below, 
but their treatment is not complete as noticed in~\cite{Aygun_Sonmez:2012}, 
and they assume additional structure on the set $X$.
We shall now propose a definition of suitable choice functions on the set of all subsets of $X$.

\section{Coherent choice functions} \label{sec:choice_functions}
Section~\ref{sec:def_choice} defines and discusses the family of choice functions
that will be considered in the paper: coherent choice functions.
Section~\ref{sec:examples_choice} presents examples of coherent choice functions.

\subsection{Definition} \label{sec:def_choice}
A {\em finite} set $X$ is given. 
The members of $X$ will be called contracts, following~\cite{Hatfield_Milgrom:2005}.
Examples of contracts are: a marriage between a specific boy and a specific girl, 
a work contract between a specific firm and a specific worker for a specific salary and specific
work conditions, a sale from a specific producer to a specific consumer of a specific good 
at a specific price deliverable in a specific place at a specific date.

Following the revealed preferences approach of~\cite{Alkan_Gale:2003}, 
we want to describe the preferences of an 
individual over sets of contracts by a choice function 
\mbox{$f : {2}^{X} \rightarrow {2}^{X}$} that, for every 
\mbox{$A \subseteq X$},
provides the preferred set $f(A)$ of contracts out of all contracts of $A$.
We shall show that, in certain circumstances, the collective preferences of a group of individuals
can also be defined by such a choice function.

The following definition encapsulates the properties we want to assume 
about the function $f$.
They seem very reasonable if we think in terms of preferences and a detailed discussion 
will be provided after Definition~\ref{def:coherent}.
It is a claim of this paper that, in all previously studied matching problems, 
the preferences of the players may be described by {\em coherent} choice functions.
\begin{definition} \label{def:coherent}
A choice function  \mbox{$f : {2}^{X} \rightarrow {2}^{X}$} is said
to be a {\em coherent} choice function iff it satisfies the following three properties,
for any \mbox{$A, B \subseteq X$}:
\begin{itemize}
\item {\bf (Contraction)} \mbox{$f(A) \subseteq A$},
\item {\bf (Irrelevance of rejected contracts - IRC)} if \mbox{$x \in A \subseteq X$} and
\mbox{$x \notin f(A)$}, then, \mbox{$f(A - \{ x \} ) \subseteq$} \mbox{$f(A)$}, and 
\item {\bf (Substitutes)} if \mbox{$ x \in B \subseteq$} \mbox{$A \subseteq X$} 
and \mbox{$x \in f( A ) $}, then 
\mbox{$x \in$} \mbox{$ f( B )$}.
\end{itemize}
\end{definition}
Note that we allow the set $f(A)$ to be empty, even when $A$ is not.
An agent may prefer no contracts to any other subset of $A$.
A completely different notion has been called by the same name in~\cite{Coherent_choice_Seidenfeld:2010}.

The three properties above have been extensively studied, among others, 
in the social choice literature, where the set $f(A)$ is understood as the set 
of all the best elements of $A$, assuming some global order on $X$.
As noticed in~\cite{Hatfield_Milgrom:2005} (footnote 4) the intuition is different when 
one is interested in matching: the sets of contracts are ordered in some way and $f(A)$ is
the best set among all subsets of $A$, assuming there is a unique such set.
Those three properties are equivalent to the properties 2.5 -- 2.7 of Blair~\cite{Blair:88}.
This set of three properties  has been studied extensively in~\cite{L:LogicsandSemantics}
in the context of nonmonotonic logics.
The present work builds on results obtained there.

Contraction expresses that the best subset of a set $A$ of contracts is a subset of $A$.
It is so obviously required that most authors do not care to mention it.

Irrelevance of rejected contracts expresses that the absence in the choice set 
of a rejected contract cannot cause the acceptance of an additional contract.
If a hospital presented with a set of doctors $A$ that includes doctor $d$ 
would reject doctor $d$, it would not have proposed to additional doctors 
if the set presented had been \mbox{$A - \{ d \}$}.
In fact,  Lemma~\ref{le:cumulativity} will show that
the set of contracts proposed stays unchanged if a rejected contract is omitted:
under the assumptions we have \mbox{$f(A - \{ x \}) =$} \mbox{$f(A)$}.
This property has been discussed and named in~\cite{Aygun_Sonmez:2012}.
Irrelevance of rejected contracts is equivalent to the Local Monotonicity 
of~\cite{L:LogicsandSemantics}.
Local Monotonicity is the model counterpart of the logical property 
named Cautious Monotonicity 
introduced in~\cite{KLMAI:89} as the proper weakening of the monotonicity property of
classical logics when considering a logic in which an additional assumption can cause 
the retraction of a previous conclusion.
In social choice theory, the importance of this property has been put in evidence 
by Aizerman and Malishevski~\cite{AizerMalish:81, Aizer:85}.
Its importance in matching theory has not been recognized so far.
Its purpose is to guarantee that the choice of the best subset of a set $A$ 
(or one of the best subsets of $A$) is done consistently over all subsets of $X$.
Suppose for example that, in a four contracts set, two subsets \mbox{$\{ a , b \}$} and
\mbox{$\{ a , b , c \} $} are both best. 
A choice function must choose one of them. 
Suppose we choose \mbox{$f(\{ a , b , c , d \}) =$} \mbox{$\{ a , b \}$}.
Choosing \mbox{$f(\{ a , b , c \}) =$} \mbox{$\{ a , b , c \}$} would be inconsistent
even though \mbox{$\{ a , b , c \}$} is best in \mbox{$\{ a , b , c \}$}.
Indeed, even though \mbox{$f(\{ a , b , c , d \}) \subseteq $}
\mbox{$\{ a , b , c \} \subseteq $} \mbox{$\{ a , b , c , d \}$}, 
\mbox{$f(\{ a , b , c \}) \not \subseteq$} \mbox{$f(\{ a , b , c , d \})$},
contradicting IRC.

Our third property is named after~\cite{Hatfield_Milgrom:2005}.
It states that, if a contract $x$ is accepted when in competition with a set $A$ of contracts,
it will be accepted when in competition with any subset $B$ of $A$.
Under the converse view it could have been termed {\em Irrelevance of added contracts}:
the addition to $B$ of new contracts (in $A$) cannot make a contract $x$ 
rejected from $B$ acceptable.
Indeed, it says that no added contract can be complementary to a rejected contract.
It is one of the remarkable intuitions of Hatfield and Milgrom that this property expresses
the fact that contracts are {\em substitutes} to one another:
if \mbox{$x \in B - f(B)$} no contract \mbox{$y \in A - B$} can be complementary to $x$
and make us have \mbox{$x \in f( B \cup \{ y \} ) $}.
The Substitutes property is equivalent to the properties (C2) and (C3) 
of Arrow's~\cite{Arrow:59} and to property $\alpha$ of Sen's~\cite{Sen:70} 
as will be shown in Lemma~\ref{le:substitutes}.
It is a kind of antimonotonicity:
if \mbox{$X \subseteq Y$}, then antimonotonicity would require:
\mbox{$f(Y) \subseteq f(X)$}, whereas Substitutes only requires
that this part of $f(Y)$ that is included in $X$ be included in $f(X)$.
It expresses the existence of some kind of coherent test by which 
the {\em preferred} elements of a set are picked up: 
the test corresponding to a superset
must be at least as demanding as the one of a subset. 
The property appears in Chernoff's~\cite{Chernoff:54}. The {\em persistence}
property of~\cite{Alkan_Gale:2003} is very similar.
In Section~\ref{sec:bilateral},
it will be shown that under many circumstances the preferences of collectives
such as hospitals or colleges may also be expressed by coherent functions.

\subsection{Examples} \label{sec:examples_choice}
The identity function is a coherent choice function. 
It describes the preferences of an agent that is satisfied by what it gets.
Note that, if $v$ is any strictly monotone valuation, the choice function $f$ defined by 
\mbox{$f(A)$} is the subset of $A$ that maximizes $v$, is the identity.

If the set of contracts $X$ is equipped with some kind of order relation essentially any
$f$ that defines $f(A)$ as the {\em best} elements of $A$, in some sense, will prove to be
a coherent choice function.
For example, if $X$ is equipped with a pre-order $R$, i.e., a reflexive and transitive relation
and $f(A)$ is the set of all elements of $A$ that are maximal in $A$ with respect to $R$,
i.e., 
\[
f(A) \: = \: \{ x \in A \mid \forall y \in A, yRx \Rightarrow xRy \}
\]
then $f$ is a coherent choice function.
The {\em responsive} valuations of~\cite{Roth:85} are also coherent choice functions.

But, in general, preferences are determined by some kind of ordering on the subsets of
$X$, not on its elements.
In particular, if every subset of $X$ is given a numerical utility, one may consider the
choice function that, given a set $A$ of contracts, chooses the subset with the highest 
utility. Note that, for this to be possible, there must be a unique subset of highest utility.
Appendix~\ref{app:coherent_gross} shows that for every v-substitutes valuation 
there is a coherent choice functions that picks out an optimal subset of any set of contracts.
There is a submodular valuation $v$ that defines
a choice function that is {\em not} coherent.
Appendix~\ref{app:Ai} contains a different characterization of coherent choice functions, 
meaningful for social choice theory but whose meaning for matching theory is not apparent, 
due to Aizerman and Malishevski. It will not be used in this paper.

It is remarkable that the {\em substitutes} notion that originates in a utility setting can be expressed at all, 
and in an elegant way at that, in a revealed preferences setting, by properties of choice functions.
The twin notion of {\em complements} does not seem to have such an expression.

\section{Bilateral many-to-many matching and collective preferences} \label{sec:bilateral}
\subsection{Contracts, individual and collective preferences} \label{sec:contracts}
In matching theory we are concerned with situations involving two parties of agents:
e.g., men and women, firms and workers, hospitals and doctors, schools and students, producers 
and consumers. A contract links two specific individuals, one from each party:
a man and a woman, a producer and a consumer and so on.
We therefore assume two {\em disjoint finite} sets $I$ and $J$ of agents.
\begin{definition} \label{def:contracts}
For every contract \mbox{$x \in X$},
\begin{itemize}
\item  $x_{I}$ denotes the agent \mbox{$i \in I$} mentioned 
in contract $x$. For any \mbox{$i \in I$}, $X_{i}$ denotes the set of contracts of which 
$i$ is part, i.e., \mbox{$X_{i} =$} \mbox{$\{ x \in X \mid x_{I} = i \}$}.
\item $x_{J}$ denotes the agent \mbox{$j \in J$} 
mentioned in contract $x$. For any \mbox{$j \in J$}, $X_{j}$ denotes the set of contracts of which 
$j$ is part, i.e., \mbox{$X_{j} =$} \mbox{$\{ x \in X \mid x_{J} = j \}$}.
\end{itemize}
\end{definition}

The preferences of agent \mbox{$ a \in I \cup J$} are represented by a choice function 
$f_{a}$ over $X_{a}$.
The collective preferences of the agents on the same side are represented by the choice functions
$f_{I}$ and $f_{J}$ on $X$ defined by
\begin{equation} \label{eq:collective}
f_{I}(A) \, = \, \bigcup_{i \in I} f_{i}(A \cap X_{i} ) , \ \  \ \ 
f_{J}(A) \, = \, \bigcup_{j \in J} f_{j}(A \cap X_{j} ).
\end{equation}
Intuitively: every agent  \mbox{$a \in I \cup J$} cares only about the contracts of $A$ he or she is mentioned in:
$A \cap X_{a}$ and from this set picks his or hers preferred subset according to his or hers preferences.
In other terms $f_{I}$ rejects a contract iff it is rejected by the agent of $I$ it belongs to,
and similarly for $J$.
Equation~(\ref{eq:collective}) expresses a {\em no externalities} assumption: 
agent $a$ is indifferent to the fate of all contracts that do not concern him or her.
Under such an assumption the choice functions $f_{I}$ and $f_{J}$ 
defined in Equation~(\ref{eq:collective}) faithfully represents 
the collective preferences of the side considered.
Even though men are competing among themselves for women, 
one may define a collective preference for the men.
But such a feat cannot be achieved by defining a collective utility for the men since
the collective preferences we want to consider must clearly reflect a very partial ordering: 
we do not want the collective to prefer the good of one man over another, 
the collective preferences must only reflect what is the joint interest of all agents 
on the same side.
One of the good news in this paper is that, in the framework of revealed preferences and
coherent choice functions, this is possible.
We shall now show that if each of the choice functions $f_{a}$ is coherent then 
the choice functions $f_{I}$ and $f_{J}$ 
defined in Equation~(\ref{eq:collective}) are coherent.

\begin{theorem} \label{the:collective}
If \mbox{$S \in \{ I , J \}$}, and if $f_{i}$ is a coherent choice function for every 
\mbox{$i \in S$}, then the choice function $f_{S}$
defined in Equation~(\ref{eq:collective}) is also coherent.
\end{theorem}
\begin{proof}
Let \mbox{$f = f_{S}$}.
For Contraction: \mbox{$\bigcup_{i \in S} f_{i}(A \cap X_{i} ) \subseteq $}
\mbox{$ \bigcup_{i \in S} A \cap X_{i} = A$}.
For Irrelevance of rejected contracts:
if \mbox{$x \in A - f(A)$}, there is a single \mbox{$j \in S$} such that 
\mbox{$x \in A \cap X_{j}$} and
\[
\bigcup_{i \in S} f_{i}(A \cap X_{i} - \{ j \} ) \, = \,
\bigcup_{i \in S - \{ j \} } f_{i}(A \cap X_{i} ) \cup f_{j} ( A \cap X_{j} - \{ x \} ) \, \subseteq \,
\bigcup_{i \in S} f_{i}(A \cap X_{i} )
\]
since $f_{j}$ is coherent.
For Substitutes, assume \mbox{$x \in B \subseteq A$} and \mbox{$x \in f( A ) $}.
There is a single \mbox{$j \in S$} such that 
\mbox{$x \in X_{j}$} and therefore \mbox{$x \in f_{j} ( A \cap X_{j} ) $}.
Since $f_{j}$ satisfies Substitutes, \mbox{$x \in f_{j} ( B \cap X_{j} ) $},
and therefore \mbox{$x \in f ( B )$}.
\end{proof}

\subsection{Preferences in the Gale-Shapley marriage problem} \label{sec:GS_preferences}
In this paper, in order to explain the algorithm presented in Section~\ref{sec:algorithm},
we shall now translate, step by step, the classical marriage problem as presented 
in~\cite{GaleShapley:marriage} into the framework proposed in this paper.
We consider a set $I$ of $n$ men and a set $J$ of $n$ women.
For Gale and Shapley the preferences of man $i$ are described by a strict total order 
$<_{i}$ on the set $J$, and the preferences of a woman $j$ by a strict total order $<_{j}$ 
on the set $I$.
Our formalization of the marriage problem uses a contract set 
\mbox{$X = I \times J$} that contains all possible pairs \mbox{$( i , j )$} 
where \mbox{$i \in I$} and \mbox{$j \in J$},
and the preferences of each agent \mbox{$i \in I$} are described by a coherent choice function 
$f_{i}$ on the set of contracts 
\mbox{$X_{i} =$} \mbox{$\{ ( i , j ) \mid j \in J \} $} in which $i$ participates.
We shall now describe the function $f_{i}$ when \mbox{$i \in I$}, i.e., $i$ is a man.
The description of the choice function of a woman is similar.
For \mbox{$A \subseteq X_{i}$}:
if $A$ is empty, $f(A)$ is empty and otherwise the set $f_{i}(A)$ is the singleton
\mbox{$\{ x \}$} where \mbox{$x = ( i , w) \in A $} is such that \mbox{$ j <_{i} w $} for any
\mbox{$j \in J$} such that \mbox{$ ( i , j ) \in A$}.
This definition expresses that $i$ prefers being married to any woman to staying single,
that he prefers one woman to more than one woman and that among the available women
his preferences are described by $<_{i}$.

Let us check that such a choice function is coherent.
Remember that \mbox{$A \subseteq X_{i}$}.
Contraction is obvious. 
For IRC, if \mbox{$x \in A - f_{i}(A) $} the woman in $x$ is not the preferred one from 
\mbox{$A$} and \mbox{$f_{i}(A - \{ x \} ) = $} $f_{i}(A)$.
For Substitutes, suppose \mbox{$A \subseteq$} \mbox{$B \subseteq M_{i}$}.
If \mbox{$(i , j ) \in A - f_{i}(A)$} then there is a \mbox{$ (i , w ) \in A$} such that 
\mbox{$w >_{i} j$} and the contract \mbox{$(i , j )$} cannot be the one chosen
in $f_{i}(B)$.

In matching hospitals and interns, couples have to be given special consideration
and have been treated by ad hoc modifications: see e.g.,
\cite{Roth_evolution:84, Roth_Peranson:99, Bronfman_Romm:2015}.
One could wonder whether, in the perspective proposed in this paper, one could model 
the preferences of a couple of doctors by a collective choice function.
The answer is that one can but this choice function is not expected to be coherent
since the two positions looked for are complementary, not substitutes.
A couple's choice function does not satisfy the Substitutes property of Definition~\ref{def:coherent}.

In the marriage problem, the collective choice function for men has the effect of removing
from a set $A$ of contracts all contracts \mbox{$(i , j )$} when there is in $A$ another
contract \mbox{$(i , j' )$} such that $i$ prefers $j'$ to $j$, and similarly for women. 
If one considers a typical situation in the matching algorithm
of Gale and Shapley in which men are gathered around women and define the set $A$ 
of all contracts as containing exactly those pairs \mbox{$( i , j )$} 
such that man $i$ is one of the men gathered around woman $j$ 
and if $f$ is the collective choice function for the women's side, then $f(A)$ is exactly
the set of pairs \mbox{$( i , j )$} such that $i$ has not been rejected by $j$ in the curent
stage of the algorithm.

\section{The agreement problem} \label{sec:agreement}
\subsection{Definition} \label{sec:definition}
Section~\ref{sec:bilateral} proposed a perspective change: don't view matching problems
as matching individual hospitals and individual doctors each of which have individual preferences
but as an agreement problem between two collective preferences 
described by two coherent choice functions.
Such a perspective subsumes many-to-many matching and allows to consider {\em rules}
or {\em constraints}, e.g., the monogamy restriction in the marriage problem, as preferences.
In our version of the marriage problem, polygamous or polyandrous marriages are considered
and it is only the fact that agents prefer a single mate to more than one mate that ensures
a one-to-one pairing.
In this framework one may easily introduce {\em soft rules} by translating them into
{\em strong preferences}.

In this Section we shall define the problem, the {\em agreement problem}, that
generalizes the matching problem, define the solution concept for this problem,
a {\em stable agreement} and show that this solution concept subsumes stable matchings.
In Section~\ref{sec:algorithm} we shall propose an algorithm to
solve all agreement problems and, then, in the remainder of this paper we shall show 
that the study of this algorithm enables us to prove that stable agreements exist 
and present a lattice structure.

We now assume a finite set $X$ of contracts 
and two {\em coherent} choice functions: \mbox{$f_{i} : {2}^{X} \longrightarrow {2}^{X}$},
for \mbox{$i = 1 , 2$}.
The coherent choice function $f_{i}$ describes the collective preferences of side $i$.
They are obtained from the individual choice functions of agents of side $i$ by Equation (\ref{eq:collective})
and are coherent by Theorem~\ref{the:collective}.
Note that, contrary to~\cite{Hatfield_Milgrom:2005}, no structure is assumed on the set $X$
of contracts: a contract is always between side $1$ and side $2$, we can forget the individuals
involved since the individuals preferences have been taken into account in the collective
preferences expressed by $f_{1}$ and $f_{2}$.
Note also that the situation considered here is that of many-to-many matching: 
both sides may enter multiple contracts. 
Roth's~\cite{Roth:84} and Blair's~\cite{Blair:88} already considered 
such a symmetric situation and the topic has attracted
a lot of attention recently, see~\cite{Hatfield_Kominers:AEJ, Hatfield_Kominers:GEB, 
Hatfield+4:2018}.
Those works assume that preferences are defined by utilities, i.e., a total order on the sets of contracts.

The basic situation is the following.
Both sides wish to agree on a set of contracts, a subset of $X$.
If a tentative agreement, i.e., a subset of $X$, is on the table,
each side may, if it wishes to, and without any need for a permission
from the other side, reject any subset of the tentative agreement 
but the permission of both sides is necessary before a contract or a number of contracts 
can be added to a tentative agreement.
This leads to the following definitions.

To help the reader's intuition we shall mention basic properties of the notions defined
in {\em Propositions} without proofs. No proposition will be used in the sequel and 
all propositions follow from the results in Section~\ref{sec:properties_algo}.

\subsection{Agreements} \label{sec:agreements}
An agreement is a set no contract from which will be taken away by any of the two parties.
\begin{definition} \label{def:agreement}
A set $A$ of contracts is said to be an {\em agreement} iff
\mbox{$A =$} \mbox{$f_{1}(A) =$} \mbox{$f_{2}(A)$}, i.e., it is a fixpoint for both
$f_{1}$ and $f_{2}$.
\end{definition}

In a sense, an agreement is a {\em small} set of contracts.
The proposition below follows from Lemma~\ref{le:sub_agreement2}.
\begin{proposition} \label{prop:sub_agreement}
The empty set is an agreement.
If $A$ is an agreement and \mbox{$B \subseteq A$}, then $B$ is an agreement.
\end{proposition}
Note that the union of two agreements is not in general an agreement.

In the marriage problem one easily sees that an agreement is a partial matching:
$A$ is an agreement iff for each man $i$ there is at most one contract of the form 
\mbox{$(i , j )$} in $A$ and for each woman $j$ there is at most one contract of the form
\mbox{$(i , j )$} in $A$. There may be unmatched agents.
The explanation is that if an agent had two or more contracts in $A$ he or she would prefer 
to throw all of them away except one as implied by the definition of the choice function
$f_{i}$ in Section~\ref{sec:GS_preferences}.

\subsection{Stable sets} \label{sec:stable}
A stable set is a set of contracts to which no addition is supported by both sides.

\begin{definition} \label{def:stable0}
A set $A$ of contracts is said to be a {\em stable set} iff 
for any set \mbox{$B \subseteq X$} of contracts such that \mbox{$A \cap B = \emptyset$}, 
\mbox{$B \neq \emptyset$}, \mbox{$B \not \subseteq f_{1}(A \cup B) \cap f_{2}(A \cup B)$}.
\end{definition}
In a sense, a stable set is a {\em large} set of contracts.
The following proposition follows straightforwardly from Lemma~\ref{le:substitutes}.
Remember $f_{1}$ and $f_{2}$ are coherent.
\begin{proposition} \label{prop:stable_superset}
The set $X$ is a stable set.
If $A$ is a stable set, then any $B$ such that \mbox{$A \subseteq B$} is a stable set.
\end{proposition}
Note that the intersection of two stable sets is not, in general, a stable set.

In the marriage problem a set $A$ is stable iff there is no pair \mbox{$(i , j)$} that is not a 
member of $A$ such that $i$ prefers $j$ to all the $w$s such that \mbox{$(i , w) \in A$}
and $j$ prefers $i$ to all the $m$'s such that \mbox{$(m, j) \in A$}.
This corresponds exactly to the notion of stability in the marriage environment.

\subsection{Stable agreements} \label{sec:stable_agreements}
We can now formalize our notion of a solution.
\begin{definition} \label{def:stable_agreement}
A set $A$ of contracts is said to be a {\em stable agreement} iff 
it is an agreement and it is a stable set.
\end{definition}

A stable agreement is a set of contracts from which no agent would like to reject a contract and no two agents 
would agree to keep an additional contract between them if offered.
In the marriage problem a stable agreement is exactly a stable matching.
We may now define the agreement problem in general.

\begin{definition} \label{def:agreement_pb}
The {\em agreement problem} is the following: given a finite set $X$ and two coherent
choice functions on $X$, $f_{1}$ and $f_{2}$, find a stable agreement, i.e., a set
\mbox{$A \subseteq X$} such that \mbox{$f_{1}(A) =$} \mbox{$f_{2}(A) = A$} and
there is no \mbox{$x \in X - A$} such that 
\mbox{$x \in f_{1}(\{ x \} \cup A) \cap f_{2}(\{ x \} \cup A)$}.
\end{definition}

\section{An algorithm for the agreement problem} \label{sec:algorithm}
\subsection{The algorithm} \label{sec:algo}
We shall now offer a general solution to the agreement problem presented in 
Definition~\ref{def:agreement_pb}.
The algorithm is an iterative one that builds a sequence of sets of contracts: 
\mbox{$Z_{j} , j \geq 0$}.
The first element of the sequence: $Z_{0}$ is taken to be equal to $X$, the set of all contracts.
The sequence is then defined by: 
\begin{equation} \label{eq:Zj}
Z_{j + 1} \: = \: ( Z_{j} - f_{1}(Z_{j} ) ) \cup 
f_{2} (f_{1} (Z_{j}) ).
\end{equation}

The initial set of contracts on the table is $Z_{0} = X$. 
Then the process proceeds in stages in the following manner: 
side $1$ picks its preferred contracts from
the pool and side $2$ picks from this set, providing the set of contracts picked by side $1$
and accepted, i.e., not rejected by side $2$.
This set is put back on the table with all contracts not picked up by side $1$ 
and the process goes into another similar stage.
Note that, in the process, all contracts rejected by side $1$ are left on the table for the next
stage, therefore it is only side $2$ that rejects contracts.
We shall see below in Lemma~\ref{le:<2f1} that indeed side $2$ 
shall like side $1$'s offers more and more, i.e., side $2$ prefers $f_{1}(Z_{j+1})$ to
$f_{1}(Z_{j})$. 

\begin{lemma} \label{le:decreasing}
If $f_{1}$ and $f_{2}$ satisfy Contraction, then the sequence $Z_{j}$ of sets of contracts
is a decreasing sequence, i.e., \mbox{$Z_{j + 1} \subseteq Z_{j}$} for any \mbox{$j \geq 0$}
and there is an index \mbox{$f \geq 0$} such that \mbox{$Z_{f} = Z_{f + 1}$}.
\end{lemma}
\begin{proof}
For the first claim:
\[
Z_{j + 1} \, = \, ( Z_{j} - f_{1}(Z_{j}) \cup f_{2} ( f_{1} (Z_{j} ) ) \, \subseteq \,
Z_{j} \cup f_{1}(Z_{j}) \, \subseteq \, Z_{j} \cup Z_{j}.
\]
The second claim holds since $X$ is finite.
\end{proof}
In the Gale-Shapley deferred acceptance algorithm: the set of matches that have not yet been
rejected (by a woman) decreases (weakly).

We shall prove in the upcoming sections that the set \mbox{$S = f_{1} ( Z_{f} )$}
is a stable agreement. We shall also show that $S$ is, in some sense, 
{\em preferred} to any stable agreement by side 1.

\subsection{A special case: Gale-Shapley} \label{sec:Gale_Shapley}
In the marriage problem, the iterative process just described parallels the Gale-Shapley
algorithm in which side $1$ is the proposing side, i.e., the side that makes the first move.
$Z_{0} = X$ is the set of all possible couples. 
The set $f_{1}(Z_{0})$ in Equation~(\ref{eq:Zj}) is the set of pairs \mbox{$(i , j)$}
where $j$ is the partner best preferred by $i$: this describes exactly the first step
of the Gale-Shapley process: agents of side $1$ gathering each around 
the agent of the other side that they prefer among all others.
The set $f_{2}(f_{1}(Z_{0}))$ describes exactly the second step of the process:
each agent of side $2$ around which gathered a non-empty set of agents from side $1$ 
chooses the one it prefers amongst those gathered, i.e., it rejects all but the prefered one.
The definition of $Z_{1}$ says that it contains all contracts of $X$ except those rejected
by side $2$ in the last step. This parallels the fact that, in the upcoming stages of the
Gale-Shapley process all pairings are still possible except those pairings that were just rejected
by the agents of side $2$.
We have seen that, when applied to the marriage situation and when \mbox{$Z_{0} = X$}
Equation~(\ref{eq:Zj}) describes exactly the Gale-Shapley process.

\subsection{Computational complexity} \label{sec:complexity}
The algorithm described in Section~\ref{sec:algo} is conceptually straightforward and elegant,
but is it efficient from the computational point of view?
The answer is: it is remarkably efficient.
Note that, by Lemma~\ref{le:decreasing} the number of steps, $f$, of the algorithm is at most the size 
of the set of contracts $X$.
The algorithm finds a subset of $X$ with certain properties and a dumb algorithm would consider each subset of $X$
and check whether it satisfies the properties or not: this requires a number of operations of the order of
$2^{\mid X \mid}$, i.e., exponential in the size of $X$.
Our algorithm requires only a linear number of operations in the size of $X$.
How large is $X$? 
Considering the discussion in Section~\ref{sec:bilateral}, the set $X$ will be something like the Cartesian product
\mbox{$I \times J \times T$} of the set of agents on one side by the set of agents on the other side by the set
of items that can be traded. 
The considerations in Section~\ref{sec:money} show that this may be multiplied by a finite set of prices, but,
all in all, the size of $X$ is polynomial in the size of the data that defines the problem.
This stands in contradistinction with algorithms for finding a competitive equilibrium:
in~\cite{NisanSegal:JET} the authors showed that finding an efficient allocation requires exponential communication.
For finding a stable agreement a polynomial number of applications of the choice functions is enough.

\section{Properties of the algorithm} \label{sec:properties_algo}
The purpose of this section is to study the properties of the sequence of $Z_{j}$'s above
to prove that $S$ is a stable agreement and study the structure of stable agreements.
To this purpose a significant mathematical apparatus is required.
Section~\ref{sec:properties_coherent} 
develops the theory of coherent choice functions.
To any coherent choice function $f$ it associates a partial pre-order $\leq_{f}$ expressing the
preferences revealed by $f$.
Section~\ref{sec:properties_iterative}
analyzes the iterative process defined in Section~\ref{sec:algorithm} in terms of the
partial pre-orders $\leq_{f_{1}}$ and 
$\leq_{f_{2}}$. 
It shows that the algorithm produces a stable
agreement.
Section~\ref{sec:ordering_agreements} studies the properties of the partial
pre-order defined in Section~\ref{sec:properties_coherent} on stable agreements.
Section~\ref{sec:extremal}
shows that the set $S$
above is a stable agreement preferred by side 1 to any stable agreement.
Section~\ref{sec:lattice} proves that the set of stable agreements has a lattice 
structure.

\subsection{Properties of coherent choice functions} \label{sec:properties_coherent}
Section~\ref{sec:first_prop} presents the elementary properties of coherent choice functions.
None of the results presented there are original. 
The most important result is C. Plott's Lemma~\ref{le:path}.
A rapid overview may be enough for a first reading.
Section~\ref{sec:pre-order} defines the partial pre-order associated with a coherent
choice function, first considered by C. Blair.
This is a fundamental tool in the sequel.
Section~\ref{sec:I} is devoted to the proof of an original technical result that will be used only 
in Section~\ref{sec:lattice} and its reading may be postponed.

\subsubsection{First properties} \label{sec:first_prop}
First, IRC is equivalent to the Local Monotonicity property 
studied in~\cite{L:LogicsandSemantics}.
\begin{lemma} \label{le:irrelevance}
The IRC property is equivalent to:
\[
{\bf Local\  Monotonicity - LM} \  \ {\rm if} \ f(A) \subseteq B \subseteq A, \  {\rm then \ }
f(B) \subseteq f(A).
\]
\end{lemma}
\begin{proof}
To show that IRC implies LM, reason by induction on the size of $A - B$.
If \mbox{$A = B$}, LM holds.
Let \mbox{$x \in A - B$}.
Assume \mbox{$f(A) \subseteq B \subseteq A$} then \mbox{$x \in A - f(A)$},
\mbox{$f(A - \{ x \} ) \subseteq$} \mbox{$f(A)$} and we have
\mbox{$f( A - \{ x \} ) \subseteq $} \mbox{$B \subseteq $} 
\mbox{$A - \{ x \}$}.
We conclude by the induction hypothesis.

Now assume LM and let \mbox{$x \in A - f(A)$}.
We have \mbox{$f(A) \subseteq $} \mbox{$A - \{ x \} \subseteq$} \mbox{$ A$} and
\mbox{$f(A - \{ x \} ) \subseteq $} $f(A)$.
\end{proof}

The next lemma shows that our definition of the Substitutes property is equivalent to the
corresponding formulation in~\cite{Sen:70, Hatfield_Milgrom:2005}
\begin{lemma} \label{le:substitutes}
A choice function $f$ satisfies the Substitutes condition iff it satisfies 
one of the three following, equivalent,properties:
\begin{enumerate}
\item \label{x:out}
for any \mbox{$B \subseteq A \subseteq X$} and for any \mbox{$x \in X$},
if \mbox{$x \in f( \{ x \} \cup A ) $}, then one has \mbox{$x \in f( \{ x \} \cup B ) $},
\item \label{original:substitutes}
if \mbox{$B \subseteq A \subseteq X$}, \mbox{$B \cap f(A) \subseteq f(B)$},
\item \label{notin:substitutes}
for any \mbox{$x \in B \subseteq A \subseteq X$}, if
\mbox{$x \notin f(B)$} then \mbox{$x \notin f( A )$}.
\end{enumerate}
\end{lemma}
\begin{proof}
\begin{enumerate}
\item Let us show that Substitutes implies property~\ref{x:out}.
If we have \mbox{$B \subseteq A$} and \mbox{$x \in f ( \{ x \} \cup A )$}, 
then \mbox{$\{ x \} \cup B \subseteq \{ x \} \cup A$} and, since \mbox{$x \in \{ x \} \cup B$},
by Substitutes we have \mbox{$x \in f ( \{ x \} \cup B ) $}.
\item Assume property~\ref{x:out} and let \mbox{$B \subseteq A$}, \mbox{$x \in B \cap f(A)$}.
We have \mbox{$\{ x \} \cup A = A$} and \mbox{$x \in f ( \{ x \} \cup A ) $}.
By property~\ref{x:out}, 
\mbox{$x \in f ( \{ x \} \cup B ) $}, but \mbox{$ \{ x \} \cup B =$} $B$.
\item We prove now that property~\ref{original:substitutes} 
implies property~\ref{notin:substitutes}.
Assume property~\ref{original:substitutes}, \mbox{$x \in B \subseteq A$} and
\mbox{$x \notin f(B)$}.
We have \mbox{$B \cap f(A) \subseteq f(B)$} and therefore we have \mbox{$x \notin B \cap f(A)$}.
But \mbox{$x \in B$} and \mbox{$x \notin f(A)$}.
\item Let us, finally, show that property~\ref{notin:substitutes} implies Substitutes.
Assume property~\ref{notin:substitutes}, \mbox{$x \in B \subseteq A$} and
\mbox{$x \in f( A )$}.
By contraposition property~\ref{notin:substitutes} implies \mbox{$x \in f( B)$}.
\end{enumerate}
\end{proof}

The next results use all three properties of Definition~\ref{def:coherent}.
\begin{lemma} \label{le:cumulativity}
If $f$ is coherent then, for any \mbox{$A, B \subseteq X$}:
\begin{itemize} \item
{\bf (Cumulativity)} if \mbox{$f(A) \subseteq B \subseteq A$}, then
\mbox{$f(B) = f(A)$}, and
\item {\bf (Idempotence)} f(f(A)) = f(A).
\end{itemize}
\end{lemma}
\begin{proof}
For Cumulativity, the assumptions imply,  by Lemma~\ref{le:irrelevance},  
that \mbox{$f(B) \subseteq$} \mbox{$ f(A)$}, 
and by Lemma~\ref{le:substitutes}
\mbox{$f(A) = $} \mbox{$B \cap f(A) \subseteq f(B)$}.
For Idempotence, note that
\mbox{$f(A) \subseteq$} \mbox{$ f(A) \subseteq A$} 
and conclude by Cumulativity.
\end{proof}

\begin{lemma} \label{le:union}
If $f$ is coherent, then \mbox{$f(A \cup B) \subseteq f(A) \cup f(B)$}.
\end{lemma}
\begin{proof}
By Contraction: \mbox{$f(A \cup B) = $} 
\mbox{$( A \cap f(A \cup B) ) \, \cup \, ( B \cap f(A \cup B ) )$}.
By Lemma~\ref{le:substitutes}, \mbox{$A \cap f(A \cup B) \subseteq f(A)$} and 
\mbox{$B \cap f(A \cup B) \subseteq f(B)$}.
\end{proof}

The sets $A$ of contracts that are fixpoints of $f$, i.e., such that \mbox{$f(A) = A$}
are particularly interesting. 
Our next result shows that any subset of a fixpoint is a fixpoint.
\begin{lemma} \label{le:sub_agreement2}
If $f$ is coherent, \mbox{$B \subseteq A \subseteq X$} and 
\mbox{$f(A) = A$}, then \mbox{$f(B) = B$}.
\end{lemma}
\begin{proof}
By assumption \mbox{$B \subseteq A$} and Lemma~\ref{le:substitutes} implies
\mbox{$B \cap f(A) \subseteq f(B)$} but \mbox{$f(A) = A$}, \mbox{$B \cap f(A) = B$}
and \mbox{$B \subseteq f(B)$}. We conclude by Contraction. 
\end{proof}

The next lemma is a powerful result of C. Plott~\cite{Plott:73} and expresses what he calls
Path Independence. It is presented there in a Social Choice context where the elements 
of $X$ are not contracts but possible social outcomes.
A preferred set of social outcomes contains all individually preferred outcomes and not,
as in this work, a set preferred to other sets, as already noticed 
in~\cite{Hatfield_Milgrom:2005} (see footnote 4).
For completeness sake a proof is provided.
\begin{lemma} [C. Plott, Path Independence] \label{le:path}
A function $f$ is coherent iff it satisfies Contraction and one of the two equivalent conditions
below, 
for any \mbox{$A, B \subseteq X$}:
\begin{enumerate}
\item \label{onef} \mbox{$f(A \cup B) = f(f(A) \cup B)$},
\item \label{twof} \mbox{$f(A \cup B) = f( f(A) \cup f(B) )$}.
\end{enumerate}
\end{lemma}
\begin{proof}
First, notice that condition~\ref{onef} implies \mbox{$f ( B \cup f(A) ) =$} 
\mbox{$f ( f( B ) \cup f(A) )$} and therefore implies condition~\ref{twof}.
Then, notice that condition~\ref{twof} implies Idempotence:
\mbox{$f(A) =$} \mbox{$f(A \cup A) =$}
\mbox{$f( f(A) \cup f(A) ) = $} \mbox{$f ( f ( A ) )$}.
Therefore it implies that \mbox{$f ( f(A) \cup f(B) ) =$}
\mbox{$f ( f(f(A)) \cup f(B) ) =$} \mbox{$f(f(A) \cup B)$}
and therefore implies condition~\ref{onef}.

Suppose now that $f$ is coherent. 
By Lemma~\ref{le:union} and Contraction
\mbox{$f( A \cup B ) \subseteq$} \mbox{$ f(A) \cup f(B) \subseteq$} 
\mbox{$f(A) \cup B \subseteq$} \mbox{$A \cup B$}.
By Cumulativity in Lemma~\ref{le:cumulativity} we conclude that
\mbox{$f( f(A) \cup B) =$} \mbox{$f ( A \cup B )$}.
We have shown that $f$ satisfies condition~\ref{onef}.

Let now $f$ satisfy Contraction and condition~\ref{onef}.
Assume \mbox{$f(A) \subseteq$} \mbox{$B \subseteq$} $A$.
We have 
\[
f ( B ) \, = \, f ( f(A) \cup B ) \, = \, f (A \cup B) \, = \, f(A)
\]
and Local Monotonicity is satisfied.
By Lemma~\ref{le:irrelevance} we conclude that $f$ satisfies Irrelevance of rejected contracts.
If, now, \mbox{$ B \subseteq A$} we have
\[
f(A) \, = \, f( B \cup (A - B) ) \, = \, f ( f(B) \cup (A - B) ) \subseteq f(B) \cup ( A - B ).
\]
We see that \mbox{$B \cap f(A) \subseteq$} $f(B)$.
By Lemma~\ref{le:substitutes} we have shown that $f$ satisfies Substitutes.
We have shown that $f$ is coherent.
\end{proof}

\subsubsection{The partial pre-order induced by a coherent choice function} 
\label{sec:pre-order}
We shall now show that any coherent choice function induces a partial pre-order 
$\leq_{f}$ on the set $2^{X}$.

The choice function $f$ naturally defines a {\em preference} relation between sets
of contracts. A set $A$ is preferred by $f$ to a set $B$ if the best subset of $A \cup B$ 
is $f(A)$, i.e., if the addition of $B$ to $A$ does not make the agent change her
mind in any way: she will stick to $f(A)$.
This relation has been considered by Blair~\cite{Blair:88} in his Definition~4.1.

\begin{definition} \label{def:leqf}
Let $f$ be a coherent choice function and \mbox{$A , B \subseteq X$}.
We shall say that, from the point of view of $f$, $A$ is preferable to $B$ and write 
\mbox{$B \leq_{f} A$} iff \mbox{$f(A \cup B) = f(A)$}.
We shall say that $f$ is indifferent between $A$ and $B$ and write 
\mbox{$A \sim_{f} B$} iff \mbox{$B \leq_{f} A$} and \mbox{$A \leq_{f} B$}.
\end{definition}

The reader should note that the relation $\leq_{f}$ is a very demanding one: 
\mbox{$B \leq_{f} A$}
requires that the contracts of $B$ do not provide any advantage over those of $A$:
if $B$ is available in addition to $A$ the elements of $B$ that are not in $A$ will not
be used at all.
This demanding character of $\leq_{f}$ implies that a claim such as the claim in
Theorem~\ref{the:final} that any stable agreement $A$ is less preferred 
than a specific set $S$ is a powerful claim.
If $f$ is defined by a v-substitutes valuation $v$ as in 
Appendix~\ref{app:coherent_gross}
the partial pre-order $\leq_{f}$ is not the total pre-order defined by the valuation function.
True, if \mbox{$B \leq_{f} A$} one has \mbox{$v(B) \leq v(A)$} but the converse does 
not hold. 

From now on, all choice functions considered are assumed to be coherent and the
assumption that $f$ is coherent will not be mentioned explicitly.

In the marriage problem, a man $i$ prefers a set \mbox{$A \subseteq X_{i}$} of contracts
to a set \mbox{$B \subseteq X_{i}$} iff the most preferred woman in $A$ is (weakly) preferred
to the most preferred woman in $B$, the empty set being less preferred than any set.
Note that, in this specific situation, the $\leq_{i}$ relation for man $i$ is a total order. 
This is not the case in general.

The following provides equivalent definitions.
\begin{lemma} \label{le:equiv_leq}
\mbox{$B \leq_{f} A$} iff \mbox{$f(A \cup B) \subseteq f(A) $} iff 
\mbox{$f(A \cup B) \subseteq A$}.
\end{lemma}
\begin{proof}
The only parts are obvious.
Suppose now that \mbox{$f(A \cup B) \subseteq A$}.
We have \mbox{$f(A \cup B) \subseteq A \subseteq A \cup B$} and by Cumulativity we have
\mbox{$f(A) = f(A \cup B)$}.
\end{proof}

We proceed to study the properties of this binary relation.
\begin{lemma} \label{le:subset_reflex}
For any \mbox{$B \subseteq A \subseteq X$}:
\begin{enumerate}
\item \mbox{$B \leq_{f} A$},
\item the relation $\leq_{f}$ is reflexive,
\item \label{simf} \mbox{$A \sim_{f} f(A)$}.
\end{enumerate}
\end{lemma}
\begin{proof}
By Definition~\ref{def:leqf}. By the previous claim. By Lemma~\ref{le:path}.
\end{proof}

\begin{lemma} \label{le:antisym}
\mbox{$A \sim_{f} B$} iff \mbox{$f(A) = f(B)$} and therefore the relation
$\sim_{f}$ is an equivalence relation.
\end{lemma}
\begin{proof}
If \mbox{$A \sim_{f} B$} we have \mbox{$f(A \cup B) = f(A)$} and 
\mbox{$f(A \cup B) = f(B)$}. 
We conclude that \mbox{$f(A) = f(B)$}.

Assume that \mbox{$f(A) = f(B)$}.
By Lemma~\ref{le:union} we have \mbox{$f(A \cup B) \subseteq$} \mbox{$ f(A) \cup f(B) = $}
\mbox{$f(A) =$} $f(B)$. Conclude by Lemma~\ref{le:equiv_leq}.
\end{proof}

Note that $\leq_{f}$ does not satisfy the antisymmetry property: 
\mbox{$B \leq_{f} A$} and \mbox{$A \leq_{f} B$} does not imply \mbox{$A = B$},
and, therefore, is not a partial order.

\begin{lemma} \label{le:transitive}
If \mbox{$C \leq_{f} B$} and \mbox{$B \leq_{f} A$} then \mbox{$C \leq_{f} A$}, i.e.,
 the relation $\leq_{f}$ is transitive.
\end{lemma}
\begin{proof}
By Path Independence (Lemma~\ref{le:path}) and our assumptions:
\[
f(C \cup A) \, = \, f(C \cup f(A)) \, = \, f(C \cup f( B \cup A)) \, = \, f(C \cup B \cup A) \, = \, 
\]
\[
f( f(C \cup B) \cup A) \, = \, f( f(B) \cup A) \, = \,
f (B \cup A) \, =\, f(A).
\]
\end{proof}

We see that, if $f$ is coherent, the relation $\leq_{f}$ is a partial pre-order, also called
partial quasi-order.

The following lemma extends the Substitutes property to the case \mbox{$B \leq_{f} A$},
instead of \mbox{$B \subseteq A$}.
\begin{lemma} \label{le:subst_<}
Let \mbox{$A, B \subseteq X$} be such that \mbox{$B \leq_{f} A$}.
If \mbox{$x \in f(\{ x \} \cup A ) $} then \mbox{$x \in$} \mbox{$ f( \{ x \} \cup B ) $}.
\end{lemma}
\begin{proof}
By Path Independence, Definition~\ref{def:leqf}, and Path Independence again:
\[
f( \{ x \} \cup A \cup B) \, = \, f( \{ x \} \cup f(A \cup B) ) \, = \,
f ( \{ x \} \cup f(A)) \, = \, f(\{ x \} \cup A).
\]
But, by Substitutes, \mbox{$x \in f( \{ x \} \cup A \cup B)$} implies 
\mbox{$x \in f( \{ x \} \cup B ) $}.
\end{proof}

\subsubsection{The operation $I_{f}$} \label{sec:I}
The following will prove useful in Section~\ref{sec:lattice}
\begin{definition} \label{def:I}
For any \mbox{$A \subseteq X$} we define
\begin{equation} \label{eq:I}
I_{f}(A) \, = \, A \cup \{ x \in X - A \mid x \notin f(\{ x \} \cup A ) \}.
\end{equation}
\end{definition}

The properties of the operation $I$ are described below.
Note, in particular, item~\ref{A<B} that expresses the preference relation $\leq_{f}$
in terms of set containment and the operation $I_{f}$.
\begin{lemma} \label{le:ABI}
For any \mbox{$A , B \subseteq X$}
\begin{enumerate}
\item if \mbox{$A \subseteq B$} then \mbox{$I_{f}(A) \subseteq I_{f}(B)$},
\item \label{IAA} \mbox{$f(I_{f}(A) ) = f(A)$},
\item \label{IfIf} \mbox{$I_{f} ( I_{f} ( A ) ) = I_{f} ( A )$},
\item \label{sim} \mbox{$I_{f} (A) \sim_{f} A$},
\item \label{IAFA} \mbox{$I_{f} ( f ( A ) ) = I_{f} ( A)$},
\item \label{AIBAB} \mbox{$f(A \cup I_{f}(B) ) = f ( A \cup B ) $},
\item \label{IAcapB} \mbox{$I_{f} ( A \cap B ) \subseteq$} \mbox{$I_{f}(A) \cap I_{f}(B)$},
\item \label{A<B} \mbox{$A \leq_{f} B$} iff \mbox{$A \subseteq I_{f}(B)$}.
\end{enumerate}
\end{lemma}
\begin{proof}
\begin{enumerate}
\item Let \mbox{$A \subseteq B$}.
For any \mbox{$x \in X - B$}, by Lemma~\ref{le:substitutes},
\mbox{$x \notin f ( \{ x \} \cup A ) $} implies \mbox{$x \notin f( \{ x \} \cup B ) $}.
Therefore 
\[
C \, = \, \{ x \in X - A \mid x \notin f(\{ x \} \cup A )  \} \, \subseteq \,
\]
\[
( B - A ) \cup \{ x \in X - B \mid x \notin f( \{ x \} \cup B ) \} \, = \, D
\]
and 
\[
I_{f} (A) \, = \, A \cup C \, \subseteq \, B \cup D \, = \, I_{f}(B).
\]
\item Let \mbox{$x \in f(I_{f}(A) ) - A $}.
Since \mbox{$A \subseteq I_{f}(A)$}, by Lemma~\ref{le:substitutes}, 
\mbox{$x \in f(\{ x \} \cup A)$}, contradicting \mbox{$x \in I_{f}(A)$}.
We see that \mbox{$f(I_{f}(A) ) \subseteq A$} and we conclude by Cumulativity.
\item By Lemma~\ref{le:union}, Contraction and item~\ref{IAA} above we have
\[
f(\{ x \} \cup I_{f}(A) ) \, \subseteq \{ x \} \cup f ( I_{f} ( A ) ) \, = \,
\{ x \} \cup f (A ) \, \subseteq \, \{ x \} \cup A \, \subseteq \{ x \} \cup I_{f}(A).
\]
By Lemma~\ref{le:cumulativity} we conclude that
\mbox{$f(\{ x \} \cup I_{f}(A) ) =$} \mbox{$f ( \{ x \} \cup A ) $}.
Therefore for any \mbox{$x \in I_{f} ( I_{f} ( A ) ) - I_{f} ( A ) $}, we have 
\mbox{$x \notin f( \{ x \} \cup A ) $} and \mbox{$x \in I_{f}(A)$}, a contradiction.
We have shown that \mbox{$ I_{f} ( I_{f} ( A ) ) \subseteq$} \mbox{$I_{f}(A)$}.
By definition \mbox{$I_{f}(A) \subseteq $} \mbox{$ I_{f} ( I_{f} ( A ) ) $}.
\item By item~\ref{IAA} and Lemma~\ref{le:antisym}.
\item We distinguish three cases.
First, for any \mbox{$x \in f(A)$}, \mbox{$x \in I_{f}(f(A))$} and 
\mbox{$x \in I_{f}(A)$}.
Secondly, for any \mbox{$x \in A - f(A)$}, \mbox{$x \in I_{f}(A)$}.
Let us show that \mbox{$x \in I_{f}(f ( A ) )$}.
By Path Equivalence, for any \mbox{$x \in A - f(A)$}, 
\mbox{$x \in I_{f}(f(A))$} iff
\mbox{$x \notin f(\{ x \} \cup f(A)) =$} \mbox{$f( \{ x \} \cup A ) =$}
\mbox{$f(A)$}. 
We see that \mbox{$x \in$} \mbox{$ I_{f}(f(A) )$}.
Thirdly, for any \mbox{$x \in X - A$}, \mbox{$x \in I_{f}( f(A)))$} iff
\mbox{$x \notin f( \{ x \} \cup f(A))$} and \mbox{$x \in$} \mbox{$I_{f}(A)$} iff
\mbox{$x \notin f( \{ x \} \cup A)$} and the two conditions are equivalent,
by Path Equivalence.
\item By Lemma~\ref{le:union}, item~\ref{IAA} above and Contraction
\[
f( A \cup I_{f}(B) ) \, \subseteq \, f(A) \cup f (I_{f}(B)) \, = \,
f( A )  \cup f(B) \, \subseteq \, A \cup B \, \subseteq \,
A \cup I_{f} ( B ).
\]
We conclude by Cumulativity.
\item Let \mbox{$x \in I_{f} ( A \cap B )$}.
If, on one hand, \mbox{$x \in A \cap B$} then clearly \mbox{$x \in I_{f} ( A \cap B )$}.
If, on the other hand, \mbox{$x \in$} \mbox{$ X - A \cap B$} then 
\mbox{$x \notin f ( \{ x \} \cup A \cap B ) $},
therefore, by Lemma~\ref{le:substitutes}, 
\mbox{$x \notin f ( \{ x \} \cup A ) $} and \mbox{$x \notin f ( \{ x \} \cup B ) $}.
Whether \mbox{$x \in A$} or \mbox{$x \in X - A$}, \mbox{$x \in I_{f}(A)$}
and similarly \mbox{$x \in I_{f}(B)$}.
\item We use Lemma~\ref{le:equiv_leq}. Assume \mbox{$f(A \cup B) \subseteq B$} and \mbox{$x \in A$}.
We have
\[
f (A \cup B ) \, \subseteq \, \{ x \} \cup B \, \subseteq A \cup B
\]
and by Lemma~\ref{le:cumulativity} also \mbox{$f( \{ x \} \cup B ) =$}
\mbox{$f( A \cup B )$}.
If \mbox{$x \in B$} then \mbox{$x \in I_{f}(B)$}, but if
\mbox{$x \in A - B$}, we have
\mbox{$x \notin f(\{ x \} \cup B ) $} and therefore \mbox{$x \in I_{f}(B)$}.
We have shown that \mbox{$A \subseteq I_{f}(B)$}.
Assume now that \mbox{$A \subseteq I_{f}(B)$}. By Lemma~\ref{le:subset_reflex}, then,
and by item~\ref{sim} above we have \mbox{$A \leq_{f}$} \mbox{$ I_{f}(B) \sim_{f} B$} 
and we conclude by Lemma~\ref{le:transitive}.
\end{enumerate}
\end{proof}

\subsection{Properties of the iterative process} \label{sec:properties_iterative}
We may already see that the end product of our iterative process, the set 
\mbox{$S = f_{1}(Z_{f} )$} is an agreement.

\begin{lemma} \label{le:S_agreement}
The set $S$ is an agreement.
\end{lemma}
\begin{proof}
By Lemma~\ref{le:cumulativity}, $f_{1}$ is idempotent and 
\mbox{$f_{1}(S) =$} \mbox{$f_{1} ( f_{1} ( Z_{f} ) ) =$} \mbox{$f_{1} ( Z_{f} ) =$} $S$.
By construction \mbox{$Z_{f} =$} \mbox{$Z_{f + 1} =$} 
\mbox{$ ( Z_{f} - f_{1}( Z _{f} ) ) \cup f_{2} ( f_{1} ( Z_{f} ) ) $}
and therefore \mbox{$f_{1}( Z_{f} ) \subseteq$} \mbox{$f_{2} ( f_{1} ( Z_{f} ) )$}
and, by Contraction we conclude that \mbox{$f_{1}( Z_{f} ) =$} \mbox{$f_{2} ( f_{1} ( Z_{f} ) )$}
i.e.,
\mbox{$S =$} \mbox{$f_{2}(S)$}.
\end{proof}

We shall now proceed step by step in the analysis of the iterative process described in
Equation~(\ref{eq:Zj}) and gather all the results in Theorem~\ref{the:final} below.
We have seen in Lemma~\ref{le:decreasing} that the sequence $Z_{j}$ is decreasing.

An offer of side $1$, i.e., a contract in $f_{1}(Z_{j})$, that has been accepted by side $2$ 
will always be offered again by side $1$.
\begin{lemma} \label{le:persistence}
For any $j$, \mbox{$j \geq 0$}, 
\[
f_{2}(f_{1}(Z_{j})) \, \subseteq \, f_{1}(Z_{j + 1}).
\]
\end{lemma}
\begin{proof}
Since \mbox{$Z_{j+1} \subseteq Z_{j}$}, by Lemma~\ref{le:substitutes}, 
we have \mbox{$Z_{j + 1} \cap f_{1}(Z_{j}) \subseteq$}
\mbox{$f_{1}(Z_{j + 1})$}.
But, \mbox{$Z_{j + 1} \cap f_{1}(Z_{j}) =$}
\mbox{$f_{2} ( f_{1} ( Z_{j} ) )$} 
by Equation~(\ref{eq:Zj}).
\end{proof}
In the Gale-Shapley deferred acceptance algorithm: a man chosen by a woman at stage $j$,
will still be available for her at stage $j + 1$.

The result below is a central part of our analysis.
From the point of view of side 2, the offers of side 1 get better and better.
\begin{lemma} \label{le:<2f1}
For any \mbox{$ j \geq 0 $}, \mbox{$f_{1} ( Z_{j} ) \leq_{2} f_{1} ( Z_{j + 1 } )$}.
\end{lemma}
\begin{proof}
By Lemma~\ref{le:equiv_leq}, it is enough to show 
\mbox{$f_{2} ( f_{1} ( Z_{j} ) \cup f_{1} ( Z_{j + 1 } ) ) \subseteq$}
\mbox{$f_{1} ( Z_{j + 1 } )$}.
By Lemma~\ref{le:union}, we have:
\mbox{$f_{2} ( f_{1} ( Z_{j} ) \cup f_{1} ( Z_{j + 1 } ) ) \subseteq$}
\mbox{$f_{2} ( f_{1} (Z_{j} ) ) \cup f_{2} ( f_{1} ( Z_{j + 1 } ) )$}.
We conclude by Lemma~\ref{le:persistence} and Contraction.

The following is an alternative proof.
By Lemma~\ref{le:subset_reflex}, \mbox{$f_{1}(Z_{i}) \sim_{2}$} \mbox{$  f_{2}(f_{1}(Z_{i}))$} and
\mbox{$f_{2}(f_{1}(Z_{i})) \leq_{2}$} \mbox{$f_{1}(Z_{i + 1} ) $}, by Lemma~\ref{le:persistence}.
We conclude by Lemma~\ref{le:transitive}.
\end{proof}
In the Gale-Shapley algorithm: after stage $j + 1$ a woman is either left in the same situation she was after stage $j$ or she has a better (for her) mate.

Our next result is central towards showing that $S$ is a stable set.
\begin{lemma} \label{le:notf2}
For any \mbox{$ j \geq 0$}, if \mbox{$x \in Z_{j} - Z_{j + 1}$},
then, for any \mbox{$k > j$}, one has
\mbox{$x \notin f_{2} ( \{ x \} \cup f_{1} ( Z_{k} ) )$} and in particular
\mbox{$x \notin f_{2} ( \{ x \} \cup S )$}.
\end{lemma}
\begin{proof}
Let \mbox{$x \in Z_{j} - Z_{j + 1}$}.
We have \mbox{$x \in f_{1}(Z_{j}) - f_{2} ( f_{1} ( Z_{j} ) ) $} and therefore
\mbox{$x \notin f_{2} ( \{ x \} \cup f_{1} ( Z_{j} ) ) $}.
But, by Lemmas~\ref{le:<2f1} and~\ref{le:transitive}, 
\mbox{$f_{1} ( Z_{j} ) \leq_{2} f_{1} ( Z_{ k } )$}.
Lemma~\ref{le:subst_<} then implies that
\mbox{$x \notin f_{2} ( \{ x \} \cup f_{1} ( Z_{ k }) ) $}.
\end{proof}

In the Gale-Shapley algorithm: a woman who has rejected a man at some point will never
wish to be matched with him at some later point.

We may now state:
\begin{lemma} \label{le:S_stable}
The set $S$ is a stable agreement.
\end{lemma}
\begin{proof}
Lemma~\ref{le:S_agreement} showed that $S$ is an agreement.
Let now \mbox{$x \in X - S$}. We distinguish two cases.
On one hand, if \mbox{$x \in Z_{f}$}, \mbox{$x \in Z_{f} - f_{1} ( Z_{f} )$}.
But, by Lemma~\ref{le:path}, \mbox{$f_{1} ( \{ x \} \cup f_{1} ( Z_{f} )  ) =$}
\mbox{$f_{1} ( \{ x \} \cup Z_{f} ) =$} \mbox{$f_{1} ( Z_{f} )$}.
We conclude that \mbox{$x \notin f_{1} ( \{ x \} \cup S ) $}.
On the other hand, if \mbox{$x \in X - Z_{f}$}, there is a \mbox{$j \geq 0$} such that
\mbox{$x \in Z_{j} - Z_{j + 1 }$}.
By Lemma~\ref{le:notf2}, then, \mbox{$x \notin f_{2} ( \{ x \} \cup S ) $}.
We have shown that $S$ is a stable set.
\end{proof}

\subsection{Ordering agreements and stable sets} \label{sec:ordering_agreements}
We need to prove some basic properties of stable sets before we can prove that $S$ is the stable agreement
preferred by side $1$.
By Definition~\ref{def:stable0} both sides cannot agree to add any contract to a stable set.
We shall show now that both sides cannot agree to add any set of contracts to a stable set.
\begin{lemma} \label{le:stable_add}
If $A$ is a stable set, then, for any \mbox{$B \subseteq X$},
\[
f_{1}(A \cup B) \cap f_{2}(A \cup B) \subseteq A.
\]
\end{lemma}
\begin{proof}
Assume $x$ \mbox{$ \in f_{1}(A \cup B) \cap f_{2}(A \cup B) - A$}.
By Contraction \mbox{$x \in B$} and by Substitutes we have 
$x$ \mbox{$\in f_{1}(A \cup \{x\}) \cap f_{2}(A \cup \{x\})$}, contradicting our assumption that $A$ is a stable set.
\end{proof}

Our next result compares the two partial pre-orders $f_{1}$ and $f_{2}$.
In the sequel addition has to be understood as addition modulo 2: \mbox{$2 + 1 = 1$}.
If side $i$ prefers a set $B$ to a stable agreement $A$,
then the other side ($i+1$) prefers A to \mbox{$f_{i}( B )$}.
\begin{lemma} \label{le:converse}
Let \mbox{$A, B \subseteq X$}.
If $A$ is a stable set and
\mbox{$A \leq_{i} B$}, then \mbox{$f_{i} (B) \leq_{i + 1} A$}.
\end{lemma}
\begin{proof}
Assume \mbox{$A \leq_{i} B$}. By Lemmas~\ref{le:subset_reflex} and~\ref{le:transitive},
\mbox{$A \leq_{i} f_{i}(B)$}, and, by Definition~\ref{def:leqf} and Idempotence,
\mbox{$f_{i} (A \cup f_{i} ( B ) ) =$} \mbox{$f_{i}(B)$}.
Since $A$ is a stable set, by Lemma~\ref{le:stable_add},
\[
f_{i}(A \cup f_{i} ( B ) ) \cap f_{i + 1}(A \cup f_{i} ( B ) ) \, \subseteq \, A.
\]
We see that
\[
f_{i} ( B ) \cap f_{i + 1}(A \cup f_{i} ( B ) ) \, \subseteq \, A,
\] 
and  therefore
\[
f_{i + 1 } (A \cup f_{i} ( B ) ) \, \subseteq \, A \cup (X - f_{i}( B ) ). 
\]
By Contraction, then, we see that 
\[
f_{i + 1 } (A \cup f_{i} ( B ) ) \, \subseteq \,
( A \cup (X - f_{i} ( B ) ) ) \cap (A \cup f_{i} ( B ) )
\]
and we conclude that
\mbox{$f_{i + 1}(A \cup f_{i} ( B ) ) \subseteq $} $A$.
We conclude, by Lemma~\ref{le:equiv_leq}, that \mbox{$f_{i} ( B ) \leq_{i + 1} A$}.
\end{proof}

It follows that $\leq_{1}$ and $\leq_{2}$ are just inverse relations on stable agreements,
which is a striking property of stable matchings.
\begin{corollary} \label{co:opp_sta}
If $A$ and $B$ are stable agreements, then \mbox{$A \leq_{i} B$} iff 
\mbox{$B \leq_{i + 1} A$}.
\end{corollary}
\begin{proof}
By Lemma~\ref{le:converse} since $A$ and $B$ are agreements.
\end{proof}

Let us also notice that, on stable agreements, the partial pre-orders $\leq_{i}$
\mbox{$i = 1 , 2$} are anti-symmetric and therefore partial orders.
\begin{theorem} \label{the:partial_order}
Among stable agreements the relations $\leq_{i}$, for
\mbox{$i = 1 , 2$}, are partial order relations.
\end{theorem}
\begin{proof}
Lemmas~\ref{le:subset_reflex} and~\ref{le:transitive} proved the relations are reflexive
and transitive.
We are left to prove that if $A$ and $B$ are stable agreements,
\mbox{$A \leq_{i} B$} and \mbox{$B \leq_{i} A$}, then \mbox{$A = B$}.
Indeed one has \mbox{$f_{i} ( A \cup B ) =$} \mbox{$f(B) = B$}
and \mbox{$f_{i} ( B \cup A ) =$} \mbox{$f(A) = A$}.
\end{proof}

\subsection{An extremal property} \label{sec:extremal}
Our next result is central towards showing that $S$ is preferred by side 1 
to any stable agreement. 
In particular, Lemma~\ref{le:AsubZf} shows that any stable agreement is a subset of $Z_{f}$.

\begin{lemma} \label{le:AsubZf}
Let $A$ be a stable set of contracts such that \mbox{$f_{2}(A) =$} $A$.
Then, \mbox{$A \subseteq Z_{i}$} for any \mbox{$i \geq 0$}.
In particular \mbox{$A \subseteq Z_{f}$}, and therefore any stable agreement is a subset of $Z_{f}$.
\end{lemma}
\begin{proof}
Assume that $A$ is a stable set and that
\mbox{$f_{2}(A) =$} $A$. 
We reason by induction on $i$.
Clearly \mbox{$A \subseteq Z_{0} = X$}.
Assume that \mbox{$A \subseteq Z_{i}$}.
By Lemma~\ref{le:subset_reflex} we have \mbox{$A \leq_{1} Z_{i}$} and
by Lemma~\ref{le:converse}, \mbox{$f_{1} ( Z_{i} ) \leq_{2} A $}, i.e.,
\mbox{$f_{2} ( f_{1} ( Z_{i} ) \cup A ) =$} \mbox{$f_{2} ( A ) = A$}.
But \mbox{$f_{1} ( Z_{i} ) \subseteq$} \mbox{$f_{1} ( Z_{i} ) \cup A $} and, 
by Lemma~\ref{le:substitutes}, we have
\mbox{$f_{1} ( Z_{i} ) \cap f_{2} ( f_{1} (Z_{i} ) \cup A ) \subseteq$}
\mbox{$f_{2} ( f_{1} ( Z_{i} ))$}.
Therefore
\mbox{$f_{1} ( Z_{i} ) \cap A \subseteq$} 
\mbox{$f_{2} ( f_{1} ( Z_{i} ) )$}.
But \mbox{$A \subseteq Z_{i} $} and we conclude that
\mbox{$A \subseteq$} \mbox{$Z_{i + 1}$}, by Equation~(\ref{eq:Zj}).
\end{proof}

We can now show that $S$ is preferred by side 1 to any stable agreement.
We summarize what we know on the set $S$.
\begin{theorem} \label{the:final}
The set $S$ of contracts is a stable agreement and for any stable agreement $A$
one has \mbox{$A \leq_{1} S$} and \mbox{$S \leq_{2} A$}.
\end{theorem}
\begin{proof}
The set $S$ is a stable agreement by Lemmas~\ref{le:S_agreement} and~\ref{le:S_stable}.
Let $A$ be any stable agreement. 
By Lemma~\ref{le:AsubZf} we have \mbox{$A \subseteq Z_{f}$}.
Therefore, by Lemma~\ref{le:subset_reflex} 
\mbox{$A \leq_{1}$} \mbox{$ Z_{f} \sim_{1}$} $S$ and by Lemma~\ref{le:transitive}
we conclude that \mbox{$A \leq_{1} S$}.
By Corollary~\ref{co:opp_sta}, \mbox{$S \leq_{2} A$}.
\end{proof}

\subsection{The lattice of stable agreements} \label{sec:lattice}
We shall now show that the set of all stable agreements exhibits a lattice structure
with respect to the partial orders $\leq_{i}$ for \mbox{$i = 1 , 2$}.
Remember that those relations are partial orders on the family of stable agreements 
by Theorem~\ref{the:partial_order}.

The result below is a direct benefit of the absence of any structural assumption on $X$.
We can easily characterize sets of contracts that are less preferred than each of two other
sets of contracts.
The operation $I$ has been defined in Equation~(\ref{eq:I}).
We write $I_{i}$ for $I_{f_{i}}$, for \mbox{$i = 1 , 2$}.
\begin{lemma} \label{le:lower-bound}
Let \mbox{$A , B, C \subseteq X$} and let \mbox{$i = 1 , 2$}. 
We have \mbox{$ A \leq_{i} B$} and \mbox{$A \leq_{i} C$} iff
\mbox{$A \subseteq$} \mbox{$ I_{i}(B) \cap I_{i}(C)$}.
\end{lemma}
\begin{proof}
By Lemma~\ref{le:ABI}, property~\ref{A<B}, we have
\mbox{$A \leq_{i} B$} iff \mbox{$A \subseteq I_{i} ( B )$}  and
\mbox{$A \leq_{i} C$} iff \mbox{$A \subseteq I_{i} ( C )$}.
\end{proof}

In the sequel we shall consider only the case \mbox{$i = 1$} and therefore we set,
for any \mbox{$B , C \subseteq X$} :
\begin{equation} \label{notation:W}
W_{B}^{C} \: = \: I_{1} ( B ) \cap I_{1} ( C ).
\end{equation}

Our goal is to find, for any pair of stable agreements $B$ and $C$, a stable agreement $D$
less preferred by side 1 than $B$ and $C$, but such that any stable agreement less preferred
by side 1 than $B$ and $C$ is also less preferred than $D$.
Since we have just seen that any set of contracts less preferred than both $B$ and $C$ by
side 1 is a subset of $W_{B}^{C}$, it is natural to consider executing the iterative process of
Section~\ref{sec:algorithm} when the set of contracts considered is $W_{B}^{C}$, not $X$.

Let;
\[
Z_{0}^{W} \, = \, W_{B}^{C}, \ Z_{j + 1 }^{W} \, = \,
(Z_{j}^{W} - f_{1}( Z_{j}^{W} ) ) \cup f_{2} ( f_{1} ( Z_{j}^{W} ) ), \ 
Z_{g}^{W} \, = \,Z_{ g + 1 }^{W}.
\]
We shall prove that \mbox{$S^{W} = f_{1} ( Z_{g}^{W} )$} satisfies our goals.
Note that since the iterative process is exactly the one described in 
Section~\ref{sec:algorithm}, we may apply Theorem~\ref{the:final}, with paying attention
to the fact that the set of contracts is now $W_{B}^{C}$, not $X$.
The following is almost immediate and does not require 
that $B$ and $C$ be stable agreements: they can be any sets.
But note that we have to deal now with two different notions of stability: stability with respect to $X$ and
stability with respect to $W_{B}^{C}$.
A set $A$ is stable with respect to $W_{B}^{C}$ iff for every \mbox{$x \in W_{B}^{C} - A$}
\mbox{$x \notin f_{1}(A \cup \{ x \} ) \cap f_{2} (A \cup \{ x \} ) $}.
A set stable with respect to $X$ is stable with respect to $W_{B}^{C}$.
In the sequel when we write {\em stable} without qualification, we mean stable with respect to $X$.

\begin{lemma} \label{le:weak_lattice}
The set $S^{W}$ is an agreement, \mbox{$S^{W} \leq_{1} B$}, 
\mbox{$S^{W} \leq_{1} C$} and \mbox{$A \leq_{1} S^{W}$} 
for any stable agreement $A$ such that \mbox{$A \leq_{1} B$} and \mbox{$A \leq_{1} C$}.
\end{lemma}
\begin{proof}
The set $S^{W}$ is an agreement by Theorem~\ref{the:final}, 
since the definition of an agreement does not depend on the base set, $X$ or $W_{B}^{C}$.
By construction \mbox{$S^{W} \subseteq W_{B}^{C}$} and therefore 
Lemma~\ref{le:lower-bound} implies that \mbox{$S^{W} \leq_{1} B$} and
\mbox{$S^{W} \leq_{1} C$}.
If $A$ is a stable agreement (with respect to $X$) such that \mbox{$A \leq_{1} B$} and \mbox{$A \leq_{1} C$},
by Lemma~\ref{le:lower-bound} it is a subset of $W_{B}^{C}$ and it is a stable set with
respect to $W_{B}^{C}$.
Therefore, by Theorem~\ref{the:final}, \mbox{$A \leq_{1} S^{W}$}.
\end{proof}

The remainder of this section is devoted to showing that, 
if $B$ and $C$ are {\em stable} agreements, then $S^{W}$ is a stable set.
Note that Theorem~\ref{the:final} implies that $S^{W}$ is stable relative to $W_{B}^{C}$,
but we still have to show that for any \mbox{$x \in X - W_{B}^{C}$},
one has \mbox{$x \notin f_{1}( \{ x \} \cup S^{W} ) \cap f_{2} ( \{ x \} \cup S^{W} ) $}.

In the Gale-Shapley marriage situation, the greatest lower bound to stable matches $B$
and $C$ is given by \mbox{$f_{2} ( B \cup C )$}.
The consideration of this set will help even in the much more general situation we are faced with.
\begin{lemma} \label{le:f1W}
If $B$ and $C$ are stable agreements, then 
\mbox{$f_{2} ( B \cup C ) \subseteq f_{1} ( W_{B}^{C} ) $}.
\end{lemma}
\begin{proof}
We shall, first, show that \mbox{$f_{2}( B \cup C ) \subseteq W_{B}^{C}$}.
Let \mbox{$x \in f_{2} ( B \cup C )$}. 
By Contraction \mbox{$x \in B \cup C$}.
Without loss of generality, let us assume that \mbox{$x \in B$}.
By Substitutes, then, \mbox{$x \in f_{2} ( \{ x \} \cup C$ )}.
If \mbox{$ x \in C$}, we have \mbox{$x \in I_{1}(B) \cap I_{1}(C) =$} 
\mbox{$ W_{B}^{C}$} and we are through.
If \mbox{$x \notin C$}, since $C$ is a stable set, we have 
\mbox{$x \notin f_{1} ( \{ x \} \cup C ) $}, i.e., \mbox{$x \in I_{1}(C) $},
but \mbox{$x \in B \subseteq I_{1}(B)$} and therefore
\mbox{$x \in W_{B}^{C}$}.

Now, we know that 
\mbox{$f_{2} ( B \cup C ) \subseteq$} \mbox{$ W_{B}^{C} \cap (B \cup C) $}.
Since, by definition, \mbox{$W_{B}^{C} \subseteq$} \mbox{$I_{1}( B )$}, 
Lemma~\ref{le:substitutes} imply that 
\mbox{$ W_{B}^{C} \cap f_{1} ( I_{1} ( B ) ) \subseteq$} \mbox{$f_{1} ( W_{B}^{C} )$}.
But, by Lemma~\ref{le:ABI} part~\ref{IAA}, 
\mbox{$f_{1} ( I_{1} ( B ) ) =$} \mbox{$f_{1} ( B ) =$} $B$.
Therefore \mbox{$W_{B}^{C} \cap B \subseteq f_{1} ( W_{B}^{C} ) $} and similarly 
\mbox{$W_{B}^{C} \cap C \subseteq f_{1} ( W_{B}^{C} )$}.
\end{proof}

\begin{theorem} \label{the:lattice}
Let $B$ and $C$ be any stable agreements. 
The set $S^{W}$ is a stable agreement, \mbox{$S^{W} \leq_{1} B$},
\mbox{$S^{W} \leq_{1} C$}, and, for any stable agreement $A$ such that
\mbox{$A \leq_{1} B$} and \mbox{$A \leq_{1} C$}, one has
\mbox{$A \leq_{1} S^{W}$}.
\end{theorem}
\begin{proof}
Except for the stability of $S^{W}$, all claims have been proved in Lemma~\ref{le:weak_lattice}.
By Lemmas~\ref{le:subset_reflex} and~\ref{le:f1W} we have
\[ 
B \leq_{2} B \cup C \leq_{2} f_{2} ( B \cup C ) \leq_{2} f_{1} ( W_{B}^{C} ).
\]
By Lemma~\ref{le:<2f1} \mbox{$f_{1} ( W_{B}^{C} ) \leq_{2} f_{1}( Z_{g}^{W} ) =$} \mbox{$S^{W}$}.
Therefore, by Lemma~\ref{le:transitive} one has \mbox{$B \leq_{2} S^{W}$}.
Similarly \mbox{$C \leq_{2} S^{W}$}.
We distinguish two cases.
First, for any \mbox{$x \in X - W_{B}^{C}$}, we may, without loss of generality, assume that
\mbox{$x \notin I_{1}(B)$}, i.e., \mbox{$x \notin B$} and 
\mbox{$x \in f_{1}( \{ x \} \cup B )$}.
Since $B$ is a stable set, \mbox{$x \notin f_{2} ( \{ x \} \cup B )$}.
Therefore, by Lemma~\ref{le:subst_<}, since \mbox{$B \leq_{2} S^{W}$}, we see that
\mbox{$x \notin f_{2} ( \{ x \} \cup S^{W} ) $}.
Secondly, for any \mbox{$x \in W_{B}^{C} - S^{W}$}, Theorem~\ref{the:final} implies that
\mbox{$x \notin f_{1}( \{ x \} \cup S^{W} ) \cap f_{2}( \{ x \} \cup S^{W} ) $}.
We have proved that $S^{W}$ is a stable set.
\end{proof}

\section{Matching with contracts with money} \label{sec:money}
One of the distinguishing features of matching among economic theory is that it does not involve money, or, 
more precisely, does not assume the existence of money.
In this section we shall show that the many-to-many matching with contracts framework, 
devoid of money, developed so far, also encompasses certain exchange economies with production.

\subsection{Preferences with money} \label{sec:pref_money}
In the framework used so far, prices, if present,  are hidden in the structure of the set $X$ 
of contracts.
Since we want to study the role of prices in markets, our first task is to take money out of the
obscurity and bring it explicitly in the framework.
We also need to bring to the fore the terms of the contracts: 
the item to be delivered and all the conditions that influence preferences since we want to compare 
the prices of contracts with the same terms.
We shall therefore assume now that every contract \mbox{$x \in X$} specifies, 
not only a member of $I$: $x_{I}$ and a member of $J$: $x_{J}$ 
but also a price \mbox{$x_{P} \in P$} and a template
\mbox{$x_{T} \in T$} where $P$ is a finite set of prices  equipped with a total order ($<$)
and $T$ is a finite set of templates.
Note that, since $X$ must be finite, we consider only a finite set of possible prices:
for example a natural number of cents less or equal to a maximal price.
A template, \mbox{$t \in T$} should be understood as containing all relevant 
information about the item to be produced, its full specification, packaging, 
time of delivery, place of delivery and so on, except the price and the identity 
of the two contracting parties.

A contract \mbox{$x\in X$} represents a contract 
between producer $x_{I}$ and consumer  $x_{J}$ for the production by $x_{I}$ 
of a specific item described in $x_{T}$ and the sale of this item to $x_{J}$ for an amount $x_{P}$, 
to be transferred from $x_{J}$ to $x_{I}$.
We shall use the notations introduced in Definition~\ref{def:contracts}: $X_{i}$, $X_{j}$ and 
even define $X_{t}$ and $X_{p}$ in a similar way.

The preferences of each agent \mbox{$a \in I \cup J$} are described
by a {\em coherent} choice function \mbox{$f_{a} : 2^{X_{a}} \rightarrow 2^{X_{a}}$}.
It seems quite natural to assume that the consumers' preferences are coherent, 
since, if we assume that a consumer's preferences are described by a $v$-substitutes valuation, 
then, by the results of Appendix~\ref{app:coherent_gross}, it can be described by a 
coherent choice function.
To assume that the producers' preferences are coherent is much more restrictive.
It seems that a factory whose operation produces two different chemicals $A$ and $B$, may well
accept a contract to sell $A$ and a contract to sell $B$ if presented with two such contracts, but
reject a contract to sell $A$ if there is no buyer for the by-product $B$. 
In other terms producers often exhibit complementarities in their preferences. 
Similarly, economies of scale in production translate in a choice function that
is not coherent.
Nevertheless the preferences of a producer with linear costs of production are described 
by a coherent choice function.
A set of contracts $X$ as above and coherent choice functions $f_{a}$ for every agent
\mbox{$a \in I \cup J$} form an economy.

Since contracts are only potential trades, not realized trades, it is reasonable to assume
that there is no shortage of contracts.
This assumption is formalized below.
\begin{definition} \label{def:no-shortage}
We say that an economy satisfies no-shortage of contracts iff
\begin{enumerate}
\item for any \mbox{$i \in I$}, \mbox{$j \in J$}, \mbox{$t \in T$}, and \mbox{$p \in P$} 
there is a contract \mbox{$x \in X$} such that \mbox{$x_{I} = i$}, \mbox{$x_{J} = j$},
\mbox{$x_{T} = t$} and \mbox{$x_{P} = p$},
\item for any contract \mbox{$x \in X$} and any stable agreement \mbox{$A \subseteq X$}
such that \mbox{$x \in A$}, there is a contract \mbox{$y \in X$} such that \mbox{$y \notin A$}
and \mbox{$y_{I} = x_{I}$}, \mbox{$y_{J} = x_{J}$}, \mbox{$y_{T} = x_{T}$} and
\mbox{$y_{P} = x_{P}$}.
\end{enumerate}
\end{definition}

In economies, since the money transfers are always from the consumer to the producer,
the preferences of the producers and the consumers concerning prices are opposed:
a producer always prefers higher prices and a consumer always prefers lower prices.
Another basic assumption is that a producer does not care to whom he or she sells
and a consumer does not care from whom he or she buys.
It is remarkable that those two basic facts can be modeled by suitable restrictions 
on the respective choice functions.
The correct formulation of the property we are looking for requires some thinking.

\begin{definition} \label{def:money_economy}
An economy that satisfies no-shortage of contracts is said to be a money-economy (m-economy) iff
for any \mbox{$A \subseteq X$}, any \mbox{$x, y \in X$} such that
\mbox{$ x_{T} = y_{T}$} and \mbox{$x_{P} < y_{P}$},
\begin{enumerate}
\item \label{producer}
if \mbox{$i =$} \mbox{$ x_{I} =$} \mbox{$ y_{I}$} and \mbox{$ x \in f_{i}(A \cap X_{i})$}, then
\mbox{$y \in f_{i}( ( A \cap X_{i} ) \cup \{ y \})$}, and
\item \label{consumer}
if \mbox{$j = x_{J} = y_{J}$} and \mbox{$ y \in f_{j}(A \cap X_{j})$}, then
\mbox{$x \in f_{j}( ( A \cap X_{j} ) \cup \{ x \})$}.
\end{enumerate}
\end{definition}
For producer $i$, if he or she has the possibility to contract with $j_{1}$ or with $j_{2}$
(they may be the same consumer) at different prices, he or she can choose to contract
with both, or with none, but if he or she chooses to contract with only one it must be for
the higher price. 
Therefore, if, when offered the bundle $A$ that contains $x$ he or she chooses to keep $x$,
it must the  case that he or she will keep the higher priced contract $y$
if such a contract is offered (in $A$ or) on top of $A$: 
either at the expense of rejecting the lower priced contract $x$, if
\mbox{$x \notin f_{i}( ( A \cap X_{i} ) \cup \{ y \}) $}, or by taking both contracts, if
\mbox{$x , y \in f_{i}( ( A \cap X_{i} ) \cup \{ y \})$}.
Similarly, mutatis mutandis, for consumers.

\subsection{Properties of stable agreements: a law of two prices} \label{sec:two_prices}
The results in previous Sections that concern matching with contracts are applicable 
to the special case of m-economies: we can define coherent collective preferences for the 
producers: denoted by $f_{I}$, coherent collective preferences for the consumers:
denoted $f_{J}$ and the notion of a stable agreement.
We can apply the results above and claim that the set of stable agreements is not empty 
and forms a lattice.

In a stable agreement any contract defines a price, so the agreement defines a price  for  every traded item.
Theorem~\ref{the:two_prices} below shows that any stable agreement in an m-economy defines
a unit price for every commodity: a commodity is a set of items that are indistinguishable by the agents.
It is the discrete, revealed preferences, version of the Law of one price much discussed
in the literature: in equilibrium all units of a given commodity are traded at the same price.
The result presented below extends significantly the {\em laws of one price}
discussed in the literature since it concerns stable agreements, not core elements.
Note also that prices appear even in the absence of utilities.
The price we have to pay for our discrete, revealed preferences framework is that we cannot exclude that two
units of a commodity be traded at slightly different prices: we can only show that if this happens those two prices are
neighboring prices: there is no price in-between.

In the framework of revealed preferences in which preferences are described by choice functions
the classical solution concepts such as competitive equilibrium, core elements or Pareto-optimal solutions 
cannot be defined in a straightforward manner since there is no obvious suitable notion of {\em better satisfied},
at least no such notion that would compare any two situations:
the binary relation of Definition~\ref{def:leqf} seems too partial to be useful.
Compared to Walrasian equilibria or core elements, stable agreements represent imperfect 
equilibria, or equilibria in imperfect markets.

We can now state a {\em Law of two prices} for m-economies: 
in a stable agreement, if different contracts that concern the same  template $t$ have different
prices those prices are neighboring prices: there is no price in-between those two prices.
If the prices are measured in cents, if the same item is sold at different prices, 
then the prices can differ only by one cent.

\begin{theorem}[Law of two prices] \label{the:two_prices}
In an m-economy that satisfies no-shortage of contracts, 
for any stable agreement $A$ and contracts \mbox{$x, y \in A$} such that
\mbox{$x_{T} =$} \mbox{$ y_{T}$},
the prices of those contracts are almost the same, i.e., 
there is no price \mbox{$p \in P$} that is strictly between $x_{P}$ and $y_{P}$.
\end{theorem}
Note that the producers $x_{I}$ and $y_{I}$ may be equal or different and similarly for 
the consumers $x_{J}$ and $y_{J}$,
and therefore Theorem~\ref{the:two_prices} describes a system-wide property, not a local one.

\begin{proof}
We reason by contradiction.
Assume that \mbox{$A \subseteq X$} is a stable agreement, that \mbox{$x , y \in A$} are
such that \mbox{$t =$} \mbox{$ x_{T} =$} \mbox{$ y_{T}$} and
\mbox{$x_{P} < p < y_{P}$}.
By no-shortage of contracts, there is a contract \mbox{$z \in X$} such that
\mbox{$z_{I} = x_{I}$}, \mbox{$z_{J} = y_{J}$}, \mbox{$z_{T} = x_{T} = y_{T}$}
and \mbox{$z_{P} = p$}.
Since $A$ is an agreement, \mbox{$A = f_{I}(A)$}, \mbox{$x \in f_{I}(A)$},
\mbox{$x \in f_{x_{I}}( A \cap X_{i} )$}. 
But \mbox{$z_{I} = x_{I}$}, \mbox{$z_{T} = x_{T}$}, \mbox{$x_{P} < z_{P}$}
and Definition~\ref{def:money_economy} implies that 
\mbox{$z \in f_{x_{I}} ( (A \cap X_{I} ) \cup \{ z \})$} and
\mbox{$z \in f_{I}(A \cup \{ z \})$}.
Similarly, we have \mbox{$y \in f_{J}(A)$}, \mbox{$y \in f_{y_{J}}(A \cap X_{J} )$}, \mbox{$z_{J} = y_{J}$},
\mbox{$z_{T} = y_{T}$}, \mbox{$z_{P} < y_{P}$} and we have
\mbox{$z \in f_{y_{J}}(A \cup \{ z \})$}, and \mbox{$z \in f_{J}(A \cup \{ z \})$}.
But \mbox{$z \in f_{I}(A \cup \{ z \}) \cap f_{J}(A \cup \{ z \})$} and $A$ is a stable set implies \mbox{$z \in A$}.

By no-shortage of contracts we conclude that there is  contract \mbox{$z' \in X$} such that
\mbox{$z' \notin A$} and \mbox{$z'_{I} = z_{I}$}, \mbox{$z'_{J} = z_{J}$},
\mbox{$z'_{T} = z_{T}$} and \mbox{$z'_{P} = z_{P}$}.
But the reasoning just above about $z$ carries over to $z'$ and we have \mbox{$z' \in A$}.
A contradiction.
\end{proof}

Note that the assumption that the agents' choice functions are coherent is not used in the proof
of Theorem~\ref{the:two_prices}. It only serves to ensure the existence of stable agreements.
The law of two prices is therefore applicable to imperfect equilibria in a large family of economies.

\section{Open problems and future work} \label{sec:open}
Below is a list of questions that need further reflection.
\begin{enumerate}
\item It is quite remarkable that substitutes agents can be exactly characterized by their choice functions.
Can other typical types of behavior be characterized this way and can the revealed preferences approach
successfully tackle other types of markets?

\item 
The study of the iterative process of Section~\ref{sec:algorithm} may be deepened.
In particular the question whether Lemma~\ref{le:f1W} can be strengthened is open.
Is the set \mbox{$f_{2}( B \cup C )$} a subset of $S^{W}$, is it equal to $S^{W}$,
as is the case in one-to-many matching?

\item
Suppose the preferences of each agent $i$ on one side are described 
by a v-substitutes valuation $v_{i}$ and the coherent choice function $f_{i}$ it defines
as in Appendix~\ref{app:coherent_gross}.
Is the collective coherent choice function of Equation~\ref{eq:collective} the choice function
defined by the v-substitutes valuation that is the convolution, i.e., 
\mbox{$\bigvee_{i} v_{i}$} as in~\cite{LLN:GEB}?

\item
Can one define useful families of coherent choice functions 
that possess short representations?

\item
One can easily formalize the restriction that an individual choice function picks a unique
best contract from any set of contracts. 
Can one similarly characterize choice functions that represent collective preferences that
aggregate such preferences? 
Can one easily prove the particular properties of one-to-one and one-to-many matchings
by considering the iterative process of Section~\ref{sec:algorithm}?

\item
Does the lattice of stable agreements enjoy additional properties? 
Under what conditions is it distributive?

\item
Is the iterative process of Section~\ref{sec:algorithm} truthful?
Can one of the sides benefit from choosing a subset that is not its preferred subset?

\item Section~\ref{sec:pref_money} considered whether it was reasonable to assume that agents,
and in particular producers preferences are described by coherent choice functions.
The study of {\em demand types} proposed in~\cite{BaldwinKemP:2019} may probably be used to classify choice
functions and study classes of agents for which a stable agreement is guaranteed to exist.

\item The stable agreements form a lattice. 
The situation of a stable agreement on the axis defined by the preference relations $\leq_{I}$ and $\leq_{J}$
describes the relative strengths of producers and consumers.
But how can one compare two stable agreements that are not related by those relations?
Could price changes move the market from one to another?

\item The algorithm presented in Section~\ref{sec:algorithm} obtains a stable agreement but does not describe
the way real markets behave in their way to some equilibrium which is not necessarily a competitive equilibrium.
One should study more realistic ways in which an unstable agreement can evolve into a stable one.
\end{enumerate}

\section{Conclusion} \label{sec:conclusion}
Many-to-many matching can be treated in a pure revealed preferences framework if
the choice functions that describe the agents' preferences are assumed to be coherent. 
The collective preferences of each side are then described too by coherent choice functions. 
A natural iterative process provides both an extremal stable agreement
and a proof of the lattice structure of stable agreements.
Two-sided markets in which agents' preferences can be described by coherent choice functions 
have a non-empty set of stable agreements and if there are  competitive equilibria they are in this set, 
but stable agreements exist in markets that lack a competitive equilibrium.
Stable agreements probably correspond to imperfect market equilibria.
Contrary to competitive equilibria, stable agreements are easy to find and only polynomial-size information 
needs to be shared to find one.

\section{Acknowledgments} \label{sec:Ack}
Discussions with Aron Matskin in a preliminary study conducted in 2004 is gratefully acknowledged.

\bibliographystyle{plain}

\appendix
\section{V-substitutes valuations and coherent functions} \label{app:coherent_gross}
We shall show that any valuation that is v-substitutes yields a coherent choice function.
We consider the most general valuations \mbox{$v : 2^{X} \rightarrow \cR$}.
The value $v(A)$ may be negative, the function $v$ is not necessarily monotone and
$v(\emptyset)$ is not necessarily equal to $0$.
The following generalizes the definition of Kelso and Crawford in~\cite{KelsoCraw:82}.
\begin{definition} \label{def:v-subst}
A valuation \mbox{$v : 2^{X} \rightarrow \cR$} is {\em v-substitutes} iff 
for any price vector \mbox{$p : X \rightarrow \cR$}, if we define
\mbox{$u_{p}(A) = v(A) - \sum_{x\in A} p(x)$}, and if we consider two price vectors
$p_{1}$ and $p_{2}$ such that \mbox{$p_{1} \leq p_{2}$} and if $A$ is a set of contracts 
that maximizes $u_{p_{1}}$ over all sets of contracts, then there is a set $A'$ that maximizes
$u_{p_{2}}$ over all sets of contracts that includes all contracts of $A$ whose prices are
the same in $p_{1}$ and $p_{2}$.
Note that prices may be negative.
\end{definition}

We want to show that any v-substitutes valuation yields a corresponding coherent choice 
function.
\begin{theorem} \label{the:char_coherent}
Let \mbox{$v : 2^{X} \rightarrow \cR$} be any v-substitutes valuation. 
There is a coherent choice function \mbox{$f : 2^{X} \rightarrow 2^{X}$} such that,
for any \mbox{$A \subseteq X$}, \mbox{$v ( f ( A ) ) \geq v ( B )$} for any
\mbox{$B \subseteq A$}.
\end{theorem}
\begin{proof}
The proof proceeds in two steps.
First, we shall show there there is a v-substitutes valuation $v'$ that is a perturbation
of $v$ that enjoys the uniqueness property: any set $A$ has a unique subset that maximizes
$v'$, and that is such that the $v'$ maximizing subset of any $A$ is a $v$ maximizing subset of $A$.

\begin{definition} \label{def:uniqueness}
A valuation \mbox{$v : 2^{X} \rightarrow \cR$} possesses the {\em uniqueness property}
iff for any \mbox{$A \subseteq X$} there is a unique \mbox{$B \subseteq A$} that maximizes
$v$ over all subsets of $A$.
\end{definition}

The second step shows that any v-substitutes valuation that enjoys the uniqueness property
yields a suitable choice function. This result is not original.
\end{proof}

\begin{lemma} \label{le:v-sub_unique}
Let $v$ be any valuation. 
One may find prices $p$  that are small enough such that, for every $A$,
the subset that maximizes $u_{p}$ is a subset that maximizes $v$ and
such that $u_{p}$ possesses the uniqueness property.
If $v$ is v-substitutes, so is $u_{p}$.
\end{lemma}
\begin{proof}
Consider the number:
\mbox{$m = \min_{A  , B \subseteq X} \{ \mid v(A) - v(B) \mid \} $} and choose
\mbox{$0 < \epsilon < \frac{m}{ \mid X \mid} $}.
Let us choose prices $p_{x}$, \mbox{$0 \leq p_{x} \leq \epsilon$} and let
\mbox{$u(A) = v(A) - \sum_{x \in A} p_{x}$}.
The function $u$ is v-substitutes, by the definition of v-substitutes.
Note that, for any \mbox{$A , B \subseteq X$}, 
\mbox{$u(A) - u(B) > v(A) - v(B) - m$} and therefore \mbox{$u(A) \leq u(B)$} implies
\mbox{$v(A) < v(B)$}. 
We conclude that the subset of $A$ that maximizes $u$ already maximized $v$.
Now, the uniqueness property is not satisfied for all choices of prices, but it is satisfied
for almost all choices of prices since, if \mbox{$A \neq B$} and \mbox{$v(A) = v(B)$}, 
\mbox{$u(A) \neq u(B)$} unless some specific linear relation holds between the prices $p_{x}$.
\end{proof}

\begin{lemma} \label{le:unique_coherent}
If \mbox{$v : 2^{X} \rightarrow \cR$} is a v-substitutes valuation that possesses the
uniqueness property, then the choice function \mbox{$f : 2^{X} \rightarrow 2^{X}$}
defined by: $f(A)$ is the unique subset of $A$ that maximizes $v$, is coherent.
\end{lemma}
\begin{proof}
Contraction and IRC follow directly from the definition.
For Substitutes, assume that \mbox{$x \in B \subseteq A$} and 
\mbox{$x \in f(A) $}.
Let $p$ be a price vector with very high prices for all contracts in \mbox{$X - A$}
and zero prices on contracts in \mbox{$A$}.
At those prices the subset of \mbox{$A$} that maximizes $u_{p}$ is
\mbox{$f(A)$} and, by assumption, it contains $x$.
Let us now consider a price vector $p'$ in which all contracts in \mbox{$A - B$}
are very high, and all other contracts keep the price they had in $p$.
Since $v$ is v-substitute, the subset of \mbox{$B$} that maximizes $u_{p'}$ contains 
every contract of $f(A)$ whose price has not been raised, in the move from $p$ to $p'$.
We see that $x$ is a member of this set.
But this set is also the subset of $B$ that maximizes $v$ among all subsets of $B$ since
elements of $B$ have zero price in $p'$: it is $f(B)$.
\end{proof}

It has been shown, in Section~\ref{sec:examples_choice}, that valuations that are not v-substitutes
may also yield coherent choice functions. 

\section{Another characterization of coherent choice functions} \label{app:Ai}
The following is a characterization of coherent choice functions due to Aizerman and 
Malishevski. The proof can be found in~\cite{AizerMalish:81}.
\begin{lemma}[Aizerman and Malishevski] 
\label{le:rationalizable}
A choice function $f$ is coherent iff there is a finite set of binary 
relations $>_{i}$ on $X$ such that, for any \mbox{$A \subseteq X$}, $f(A)$ is 
the set of all elements of $A$ that are maximal in $A$ for at least one of the $>_{i}$'s.
\end{lemma}

\section{Inexistence of a stable agreement} \label{sec:inexistence}
The following example will show that the existence of a stable agreement is not always guaranteed.
\begin{example}
Assume two agents $1$ and $2$ and two contracts $a$ and $b$.
Let 
\[
f_{1}(\emptyset) = f_{2}(\emptyset) = f_{1}(\{ b \} ) = \emptyset,
\]
\[ 
f_{1}(\{ a \} ) = f_{2}( \{ a \} ) = \{ a \},
\]
\[
f_{2}(\{ b \} ) = f_{2} ( \{ a , b \} ) = \{ b \},
\]
and \mbox{$f_{1}( \{ a , b \} ) = \{ a , b \}$}.
Note that $f_{2}$ satisfies path equivalence but $f_{1}$ does not.
The empty set is not stable since \mbox{$a \in f_{1} ( \{ a \} ) \cap f_{2}( \{ a \})$}.
The singleton \mbox{$\{ a \}$} is not stable since \mbox{$b \in f_{1}(\{ a , b \} ) \cap f_{2}(\{ a , b \} ) $}.
The singleton \mbox{$\{ b \}$} is not an agreement since \mbox{$b \not \in f_{1}( \{ b \} )$}.
The set \mbox{$\{ a , b \}$} is not an agreement since \mbox{$a \not \in f_{2}(\{ a , b \} ) $}.
\end{example}

Suppose there are two agents, a seller $s$ and a buyer $b$. 
The seller is interested in selling two objects: \mbox{$a , b \in O$}.
The set of contracts $X$ is the set \mbox{$O \times \cal R^{+}$} 
where $\cal R^{+}$ is the set of nonnegative real numbers.
The first element of a contract is the item to be sold and the second element is the price.
The seller is interested in selling both $a$ and $b$ at any price.
Therefore his/her choice function $f_{s}$ chooses, in a set \mbox{$A \subseteq X$} the contract for $a$ with the
highest price, if there is one such contract, and the contract for $b$ with the highest price, if there is one such contract. 
The function $f_{s}$ satisfies path equivalence.

\end{document}